\documentclass[10pt]{article}
\PassOptionsToPackage{numbers, compress}{natbib}
\usepackage[margin=1.1in]{geometry}
\usepackage{amsmath,amsthm}
\PassOptionsToPackage{hyphens}{url}
\usepackage{hyperref}
\hypersetup{
    colorlinks=true,       % Enable colored links
    linkcolor=red,        % Color of internal links (e.g., sections, equations)
    citecolor=blue,         % Color of citation links (e.g., \cite{})
    filecolor=magenta,     % Color of file links
    urlcolor=cyan          % Color of URLs
}

\usepackage[noend]{algorithmic}
\usepackage[algoruled,nofillcomment,algo2e]{algorithm2e}
\usepackage{hyperref}
\usepackage{url}
\usepackage{balance}
\usepackage{multirow}
\usepackage{makecell}
\usepackage{graphicx} % Required for inserting images
\usepackage{xspace}
\usepackage{booktabs}
\usepackage{tikz,pgfplots}
\usepackage{pgfplotstable}
\pgfplotsset{compat=1.7}
\usetikzlibrary{patterns}
\usepackage[caption=false,font=footnotesize]{subfig}
\usepackage{xcolor}
\usepackage{cleveref}
\usepackage{todonotes}
\newtheorem{definition}{definition}
\newtheorem{theorem}{Theorem}
\newtheorem{lemma}[theorem]{Lemma}
\newtheorem{prop}{Proposition}
\def\name{\textit{Nebula}\xspace}
\usepackage{pifont}% http://ctan.org/pkg/pifont
\newcommand{\cmark}{\ding{51}}%
\newcommand{\xmark}{\ding{55}}%
\newcommand{\threshold}{\ensuremath{\tau}\xspace}

\newcommand{\PDHE}{private distributed histogram estimation}

\newcommand{\eg}{{e.g.,}\xspace}
\newcommand{\ie}{{i.e.,}\xspace}
\newcommand{\myparagraph}[1]{\vspace{1ex}\noindent{\bf #1}}

\usepackage{tcolorbox}
\usepackage{multirow}
\usepackage[shortlabels]{enumitem}
\usepackage{url}
\usepackage[utf8]{inputenc} % allow utf-8 input

\usepackage[T1]{fontenc}    % use 8-bit T1 fonts
\usepackage{booktabs,stmaryrd}       % professional-quality tablesf
\usepackage{nicefrac}       % compact symbols for 1/2, etc.
\usepackage{microtype}      % microtypography
\usepackage{xcolor}         % colors
\usepackage{makecell}
\usepackage{subfig}

\usepackage{graphicx}
\usepackage[noend]{algorithmic}
\usepackage[algoruled,nofillcomment,algo2e]{algorithm2e}
\usepackage{placeins}
\usepackage{array}
\usepackage{dsfont}
\usepackage{bbold}
\usepackage{tikz}
\usepackage{comment}
\usetikzlibrary{positioning}
\usepackage{float} 
\usepackage{colortbl}
\def\name{\textit{Nebula}\xspace}

\begin{document}

\date{}
\title{\name: Efficient, Private and Accurate Histogram Estimation}

\author{Ali Shahin Shamsabadi$^\dag$,
    Peter Snyder$^\dag$,
    Ralph Giles$^\dag$,
    Aurélien Bellet$^\ddag$, and Hamed Haddadi$^{\dag,\diamond}$\\
    $^\dag$ Brave Software $^\ddag$ Inria, Université de Montpellier $^\diamond$ Imperial College London}

\maketitle

\begin{abstract}

\looseness=-1 We present \textit{Nebula}\footnote{Accepted at the 2025 ACM SIGSAC Conference on Computer and Communications Security (CCS’25).}, a system for differentially private histogram estimation on data distributed among clients. \textit{Nebula} allows clients to independently decide whether to participate in the system, and locally encode their data so that an untrusted server only learns data values whose multiplicity exceeds a predefined aggregation threshold, with $(\varepsilon,\delta)$ differential privacy guarantees.
Compared to existing systems, \textit{Nebula} uniquely achieves:
\textit{i)} a strict upper bound on client privacy leakage; \textit{ii)} significantly higher utility
than standard local differential privacy systems; and \textit{iii)} no requirement for trusted third-parties, multi-party computation, or trusted hardware. We provide a formal evaluation of \textit{Nebula}'s privacy, utility and efficiency guarantees, along with an empirical assessment on three real-world datasets. On the United States Census dataset, clients can submit their data in just 0.0036 seconds and 0.0016 MB (\textbf{efficient}), under strong $(\varepsilon=1,\delta=10^{-8})$ differential privacy guarantees (\textbf{private}), enabling \textit{Nebula}'s untrusted aggregation server to estimate histograms with over 88\% better utility than existing local differential privacy deployments (\textbf{accurate}). Additionally, we describe a variant that allows
clients to submit multi-dimensional data, with similar privacy, utility, and
performance. 
Finally, we provide an implementation of \textit{Nebula}.

\end{abstract}

\section{Introduction}
\label{sec:intro}

Aggregated user data allows software developers and service providers to develop, deploy, and improve their systems in various use-cases such as browser telemetry~\cite{corrigan2020privacy}, financial crime~\cite{bogdanov2016privacy}, and digital health~\cite{andrew2023anonymization}. 
However, large-scale collection of user data introduces privacy risks, as client data may contain privacy-sensitive information (\eg client preferences/interests, transactions, and medical diagnoses \cite{bharadwaj2023federated,chadha2023differentially,boneh2021lightweight,CCS:DSQGLH22,cormode2022sample,bell2022distributed,corrigan2017prio,CCS:ErlPihKor14}).

In this work, we focus on the problem of \emph{private distributed histogram estimation}, where a central server  aims to estimate the frequencies of different data values distributed among a set of clients, while providing privacy guarantees to the clients.

Several approaches to \PDHE{} have been proposed, either in published research or in deployed
systems. One such technique is \emph{threshold-aggregation},
where the server is able to learn client values if and only if sufficiently many clients contribute
the exact same value~\cite{CCS:DSQGLH22,boneh2021lightweight,li2024popstar}. Threshold-aggregation
has the benefit of providing a simple and intuitive privacy model, but generally
lacks robust, provable privacy guarantees, due to using the \emph{deterministic} notion of $K$-anonymity
for protecting the privacy of clients~\cite{narayanan2006break,holohan2017k}.

A second type of approach to \PDHE{} relies on \emph{differential privacy}
(DP)~\cite{dwork2006differential,dwork2014algorithmic} to provide formal privacy guarantees
through statistical indistinguishability. A wide range of DP-based systems for privacy-preserving data collection
have been proposed, all of which require implementers and deployers to make unappealing tradeoffs, even for state-of-the-art systems. Local DP systems~\cite{bassily2015local,wang2017locally,acharya2019hadamard,wang2019collecting} 
provide strong privacy guarantees but generally poor utility, while central DP systems~\cite{xu2013differentially,dwork2006differential,dwork2014algorithmic} provide high utility but require prohibitive levels of trust by clients.

A third category of \PDHE{} systems attempt to achieve both high utility and privacy, but do so by introducing additional costs, relying for instance on computationally expensive, novel cryptography~\cite{balle2019privacy,cheu2019distributed,erlingsson2019amplification,prochlo,cormode2022sample,bell2022distributed},  
(\eg multi-party computation, homomorphic encryption), and/or requiring multiple rounds of heavy communication between participants~\cite{bell2022distributed},
among other concerns. These systems entail computational, bandwidth and financial costs that make adoption difficult-to-impossible for all
but resource-rich organizations.

In this paper, we describe a novel system for the problem of \PDHE{}\footnote{Our system can answer related problems such as heavy
hitters and quantiles~\cite{cormode2022sample}.} that avoids the limitations and trade-offs of existing approaches, achieving simultaneously provable differential privacy guarantees, high utility, and practical efficiency. Our system, called \name, is a novel combination of the recent sample-and-threshold mechanism~\cite{cormode2022sample} with
verifiable client-side thresholding~\cite{CCS:DSQGLH22} to provide DP guarantees without the use of trusted third parties.
\name relies on two \emph{untrusted} (non-colluding) servers: a \emph{randomness server} and an \emph{aggregation server}, and in contrast to existing state-of-the-art
systems (\eg{} \cite{prochlo}), \emph{no communication between the two servers} is required.

At a high level, each client participating in \name{} begins by randomly deciding whether to contribute any data. Clients that do decide to participate locally encode their value
using a secret-sharing scheme, which prevents the server from observing uncommon values. This secret sharing process is
very cheap, requiring only a single round of oblivious communication with the untrusted randomness server, which executes a verifiable oblivious
pseudorandom function over the client's value. Participating clients then contribute their secret share to the aggregation server over an oblivious communication channel. The aggregation server then combines all received shares to recover values which have been contributed by a sufficient number of participants.

\name enforces formal differential privacy protection for all clients through three steps: \textit{i)} the uncertainty of any particular client contributing \emph{any} value; \textit{ii)} blinding the aggregation server to uncommon values through the secret-sharing mechanism (\ie thresholding); and \textit{iii)} having some clients contribute precisely defined amounts of ``dummy data'' to obscure the distribution of uncommon (\ie unrevealed) values.

In summary, we make the following contributions to
the problem of \PDHE{}:
\begin{enumerate}[leftmargin=*,topsep=-5pt,itemsep=0pt]
    \item the \textbf{design of a novel system} for conducting privacy preserving
        data aggregation under DP guarantees that achieves all of the following:
        \textit{i)} high utility, particularly when compared to other DP-based systems
        with comparable privacy guarantees; \textit{ii)} realistic trust assumptions;
        and \textit{iii)} practical efficiency in terms of computational, bandwidth, and
        financial costs.
    \item a \textbf{formal analysis} of the system's privacy, security, utility, and efficiency
        guarantees.
    \item \looseness=-1 \textbf{empirical measurements} of the system's utility and
        efficiency over several real-world datasets. We provide an implementation demo of \name as supplementary material.
\end{enumerate}

\section{Problem \& Threat Model}
\label{sec:problem_formulation}

We consider a scenario where an untrusted service provider (the aggregation server) wants to obtain a histogram over $N$ clients data, $D=\{x_i\}_{i=1}^N$, where data point $x_i$ is held by the $i$-th client. Collecting data generated by clients and publishing the histogram might introduce privacy risks such as misusing information for profit or mass surveillance purposes~\cite{corrigan2017prio} as clients' data contain privacy-sensitive information~\cite{bharadwaj2023federated,chadha2023differentially,boneh2021lightweight,CCS:DSQGLH22,cormode2022sample,bell2022distributed,corrigan2017prio,CCS:ErlPihKor14}. Therefore, the data collection procedure and published histogram must protect the privacy of clients. We aim to design a system that enables the aggregation server to construct an \textbf{accurate} and \textbf{differentially private} histogram over clients data \textbf{without trusting servers} and \textbf{without imposing high computational and communication costs}. To achieve this, we rely on an additional untrusted party: the randomness server. We assume that the aggregation and the randomness server are non-colluding, which is a common assumption~\cite{bell2022distributed,bohler2020secure,CCS:DSQGLH22} as collusion can be made infeasible or too costly via physical means~\cite{lepinksi2005collusion,alwen2008collusion} or strict legal regulations~\cite{kamara2011outsourcing}. Following the literature~\cite{bell2022distributed}, we consider honest-but-curious clients who submit their data through an anonymizing proxy.

\looseness=-1 To protect the clients' privacy, we use Differential Privacy (DP)~\cite{dwork2006differential,dwork2014algorithmic}. In particular, we design a randomized protocol $\mathcal{A}$ that outputs a histogram 
over clients data which is close to the true histogram 
of $D$ while satisfying {$(\varepsilon,\delta)$-DP}: 
$
    \text{Pr}[\mathcal{A}(D)\in S] \leq e^{\varepsilon}\text{Pr}[\mathcal{A}(D')\in S]+\delta,
$
for any subset of possible output histograms $S \in \text{Range}(\mathcal{A})$ and for any two neighboring datasets $D$ and $D'$ where $D'$ is obtained by removing one client's data from $D$. The privacy budget $\varepsilon$ upper bounds the privacy leakage in the worst possible case. The smaller the $\varepsilon$, the stronger the privacy guarantees. Setting $\delta>0$ allow to relax the privacy requirement for unlikely events. Specifically, our protocol guarantees that the final output (\ie published histogram) satisfies DP, while the aggregation server's view satisfies computational DP~\cite{mironov2009computational}, a restriction of DP to computationally bounded adversaries commonly considered in privacy-preserving secure protocols. Among other cryptographic schemes, we use a $\threshold$-out-of-$N$ secret sharing scheme~\cite{bellare2020reimagining}, $\Pi_{\threshold,N}$, built out of two standard functionalities: 1) producing a random $\threshold$-out-of-$N$ share of a private value through a probabilistic algorithm with explicit randomness received as input; 2) recovering the private value after receiving at least its $\threshold$ valid secret shares.
The randomness server generates randomness required for $\Pi_{\threshold,N}$, without seeing any plaintext clients' data and in a verifiable manner (\ie clients can verify whether the randomness server has correctly followed the protocol in zero knowledge).

To ensure \emph{efficiency} by minimizing financial costs, computational overhead, and bandwidth consumption, our system incorporates the following design principles:  1) it avoids any communication between the randomness server and the aggregation server; 2) it precludes communication between clients; and 3) it requires minimal efforts from clients, with a single single round of interaction with each server. Avoiding such communications and interactions also makes the practical deployment of non-collusion assumptions more feasible and easier to maintain.

\looseness=-1 Note that the server and clients agree on two public parameters: \textit{i)} the desired differential privacy guarantee ($\varepsilon$,$\delta$); and \textit{ii)} the security parameter $\kappa$ used for the secret sharing scheme.

\section{\name Design}
\label{sec:method}

\begin{algorithm2e*}[t!]
\algsetup{linenosize=\tiny}
\small
\DontPrintSemicolon
\SetKwComment{Comment}{{\scriptsize$\triangleright$\ }}{$\quad\quad$}
\caption{\name}\label{alg:mechanism1}
        \KwIn{$N$ clients, one randomness server, one aggregation server, Truncated Shifted Discrete Laplace distribution $\text{TSDLap}(\cdot)$, DP guarantee $(\varepsilon,\delta)$,\\ $\threshold$-out-of-$N$ secret-sharing scheme $\Pi_{\threshold,N}$, public key parameter $\text{pp}$, hash function $H(\cdot)$}
        \KwOut{Clients' submissions revealed to the aggregation server}
\BlankLine
\begin{algorithmic}[1]
\STATE $(\varepsilon_{\text{Re}},\delta_{\text{Re}}),(\varepsilon_{\text{Unre}},\delta_{\text{Unre}}) \leftarrow (\varepsilon,\delta)$ \Comment*[r]{All parties agree on sample-and-threshold (Re) and dummy-data (Unre) DP guarantees}
\STATE $p_s,\threshold \leftarrow (\varepsilon_{\text{Re}},\delta_{\text{Re}})$ \Comment*[r]{Computing sampling rate and aggregation threshold} 
\FOR{$i  = 1,\ldots, N$} 
\STATE $r_i=\textit{Client-RandomnessServer}(x_i,\text{pp},H(\cdot))$ \Comment*[r]{Oblivious and verifiable randomness generation (Algorithm~\ref{alg:client-RandomnessServer})}
\STATE $ \text{sbm}, \_ \leftarrow \textit{LocalSecretSharing}(x_i,r_i,\Pi_{\threshold,N})$ \Comment*[r]{Each client locally encrypts their data (Algorithm~\ref{alg:client-encoding})}
\STATE $z_i=\text{Random}([0,1])$ \Comment*[r]{Each client locally performs a Bernoulli test to decide whether to participate}
\IF{$z_i \leq p_s$}

\STATE Submit sbm to the aggregation server
\ENDIF
\ENDFOR
\STATE $\text{Dummy}=\textit{DummyDataCreation}(\threshold,\text{TSDLap}(\cdot),(\varepsilon_{\text{Unre}},\delta_{\text{Unre}}))$ \Comment*[r]{Dummy data creation to protect unrevealed submissions (Algorithm~\ref{alg:dummydata})}
\STATE Submit Dummy to the aggregation server
\STATE $\text{ReceivedData}=(\text{sbm}\cup\text{Dummy})$ \Comment*[r]{Received encrypted data containing indistinguishable dummy and real messages}
\STATE $\text{RecoveredData}=\textit{Aggregation}(\text{ReceivedData})$ \Comment*[r]{The aggregation server performs data aggregation and recovery (Algorithm~\ref{alg:aggregation})}
\STATE \textbf{Return} $\text{RecoveredData}$
\end{algorithmic}
\end{algorithm2e*}

\begin{algorithm2e}[t!]
\algsetup{linenosize=\tiny}
\small
\DontPrintSemicolon
\SetKwComment{Comment}{{\scriptsize$\triangleright$\ }}{$\quad\quad$}
\caption{\textit{Client-RandomnessServer}: Interaction between clients and the randomness server}\label{alg:client-RandomnessServer}
        \KwIn{A client holding a private item $x$, a randomness server, a secret key $\text{msk}$, hash function $H(\cdot)$}
        \KwOut{Randomness $r$}
\BlankLine
\begin{algorithmic}[1]    
    \STATE $h=H(x)$ \Comment*[r]{Client hashes its value} 
    \STATE $r' \leftarrow R$ \Comment*[r]{Client generates a random value}
    \STATE $b=h^{r'}$ \Comment*[r]{Client sends blinded hash to the server} 
    %\STATE $(\text{msk}, \text{mpk}) \leftarrow \text{KeyGen}(\text{pp})$  \Comment*[r]{The randomness server generates a keypair based on the public key parameter given the security parameter $\kappa$} 
    \STATE $\text{z}=b^{\text{msk}}$ \Comment*[r]{Server responds with its ZKproof}
    \STATE $w=\text{z}^{\frac{1}{r'}}$ \Comment*[r]{Client unblinds the response}
    \STATE $r=H(w,x)$ \Comment*[r]{Client obtains the randomness}
    \STATE \textbf{Return} $r$
\end{algorithmic}
\end{algorithm2e}

We present the design of our novel DP and secure system, called \name.  \name requires no communication between clients, and only requires two \emph{non-cooperating servers} (one that operates an oblivious and verifiable pseudorandom function~\cite{tyagi2022fast}, and one that aggregates and learns threshold-meeting values from clients).

At a high level, \name (Algorithm~\ref{alg:mechanism1}) works as follows: \textit{i)} Each client independent of other clients obliviously communicates with the randomness server and encrypts its data; \textit{ii)} Each client performs a Bernoulli test on whether to participate: with probability $p_s$ it participates and sends its encrypted data to the server, otherwise it abstains; \textit{iii)} A randomly selected client submits dummy data by creating groups of dummy data for each possible group of unrevealed items in $\{1, ..., \tau-1\}$ to bound the information that the aggregation server might learn from unrevealed submissions; \textit{iv)} The aggregation server receives real submissions and dummy data, and performs the decoding such that it learns aggregate submissions shared by at least $\tau$ clients in the sampled set. Note that the system is designed such that dummy data does not impact the correctness and utility of the aggregations (see Section~\ref{sec:dummydataInject}). In the rest of this section, we describe each of these steps in detail.

\subsection{Oblivious \& Verifiable Randomness} 
%Generation} 

Each client starts by sampling randomness $r$ from the randomness server that runs a Verifiable Oblivious PseudoRandom Function (VOPRF)~\cite{CCS:DSQGLH22} that adheres to the standard ideal functionality~\cite{albrecht2021round} with security guarantees proven in the Universal Composability framework~\cite{jarecki2014round}. This VOPRF construction allows clients that contribute the same original value to consistently receive the same randomness $r$, while ensuring that:
i) the randomness server does not learn the clients' values;
ii) the server cannot detect when multiple clients share the same input;
iii) clients do not learn the server's PRF keys; and
iv) no communication is required between clients.

\myparagraph{Server-side setup.} The randomness server initializes the VOPRF by generating public cryptographic parameters $\text{pp} \leftarrow \text{VOPRF.setup}(1^\kappa)$ given the security parameter $\kappa$. The randomness server then generates a keypair $(\text{msk}, \text{mpk}) \leftarrow \text{KeyGen}(\text{pp})$, consisting of a secret key $\text{msk}$ and a public key $\text{mpk}$, using a Key Generation algorithm $\text{KeyGen}$ parameterized by $\text{pp}$. 

Once the VOPRF setup is complete, each client interacts with the randomness server to obtain its randomness $r$, as described in Algorithm~\ref{alg:client-RandomnessServer} and outlined below.

\myparagraph{Client-side requests.} Each client produces a request using their input data as follows.
The client first samples a local blinding factor $r'$, then computes a blinded representation of their original data value $x$ as $b=H(x)^{r'}$. The client then sends the blinded value $b$ to the untrusted randomness server.

\myparagraph{Server-side responses.} Upon receiving $b$, the randomness server evaluates the VOPRF function by computing $z=b^{\text{msk}}$ using their secret key $\text{msk}$. The randomness server returns $z$ to the client. 

Finally, the client unblinds the received response($w=z^{1/{r'}}$) and derives the final pseudorandom output as ($r=H(w,x)$).

\begin{algorithm2e*}[t!]
\algsetup{linenosize=\tiny}
\small
\DontPrintSemicolon
\SetKwComment{Comment}{{\scriptsize$\triangleright$\ }}{$\quad\quad$}
\caption{\emph{LocalSecretSharing}: Client data encoding}\label{alg:client-encoding}
        \KwIn{A client holding a data point $x$, randomness $r$, $\threshold$-out-of-$N$ secret-sharing scheme $\Pi_{\threshold,N}$}
        \KwOut{Encoded data $\text{sbm}$}
\BlankLine
\begin{algorithmic}[1]
%\STATE \rebuttal{Parse $r$ into three parts so that $r_a\|r_b\|r_c = r$}
\STATE $r_1,r_2,r_3 \gets H(r_a\|1), H(r_b\|2), H(r_c\|3)$ such that $r_a\|r_b\|r_c = r$ \Comment*[r]{Parsing the randomness to three random values}
%\STATE $ \{r_1, r_2, r_3\}= \text{RandomOracleHash}(r)$ \Comment*[r]{Parsing the randomness to three random values}
\STATE $ \text{Key} = \text{PseudorandomGenerator}(r_1)$ \Comment*[r]{Deriving a symmetric key using pseudorandom generator with $r_1$ as the seed}
\STATE $\mathbf{c} \leftarrow {\text{Enc}}({\text{Key},\mathbf{x}})$ \Comment*[r]{Encrypting the data and generating a ciphertext}
\STATE $ t \leftarrow r_3$ \Comment*[r]{Generating a tag for the data using $r_3$}
\STATE $ s = \Pi_{\threshold,N}(r_1;r_2)$ \Comment*[r]{Constructing a secret-share of the  random value $r_1$ used for deriving the encryption key}
\STATE ${\text{sbm} \leftarrow (\mathbf{c}, s, t)}$ \Comment*[r]{Creating a tagged submission for the data}
\STATE \textbf{Return} sbm, Key
\end{algorithmic}
\end{algorithm2e*}

\begin{algorithm2e}[t!]
\algsetup{linenosize=\tiny}
\small
\DontPrintSemicolon
\SetKwComment{Comment}{{\scriptsize$\triangleright$\ }}{$\quad\quad$}
\caption{\textit{DummyDataCreation}: Create groups of dummy data}\label{alg:dummy}
        \KwIn{A public thresholding value $\threshold$, Truncated Shifted Discrete Laplace distribution $\text{TSDLap}(\cdot)$, DP guarantees $(\varepsilon_{\text{Unre}},\delta_{\text{Unre}})$}
        \KwOut{A set of dummy data}
\BlankLine
\begin{algorithmic}[1] 
    \STATE $\text{Dummy}=\{\}$ \Comment*[r]{The set containing groups of dummy data}
    \STATE Select a client for creating dummy data
    \FOR{$i = 1,\ldots,\threshold-1$}
    \STATE $\alpha \leftarrow \text{TSDLap}(\lambda=2/\varepsilon_{\text{Unre}},t=2+2/\varepsilon_{\text{Unre}}\log(2/\delta_{\text{Unre}}))$ 
    \STATE $\{\text{t}_j\}_{j=1}^\alpha=\text{UniqueTagGenerator}(\alpha)$ \Comment*[r]{Unique tags}    
    \FOR{$j  = 1,\ldots, \alpha$}
    \STATE $s_j=\{(c_j,s_j,\text{t}_j)^i\}$ \Comment*[r]{The client creates a set containing $i$ zero-value items with the same unique tag}
    \STATE $\text{Dummy}.\text{append}(s_j)$
    \ENDFOR
    \ENDFOR
    \STATE \textbf{Return} $\text{Dummy}$
\end{algorithmic}\label{alg:dummydata}
\end{algorithm2e}

\subsection{Local Data Preparation and Submission}
To secret-share the data (Algorithm~\ref{alg:client-encoding}), the client parses $r$ into $\{{r_1},{r_2},{r_3}\}$ using a random oracle model hash function. Each of these three randomness components are used for different purposes.

$r_1$ is used to seed a PR generator function and derive a {${\text{Key}}$} for a symmetric encryption scheme which satisfies IND-CPA security and consists of two algorithms: i) encryption: producing ciphertext of a data with key; and ii) decryption: outputting a data given its ciphertext under the key. Using the encryption algorithm, we obtain the encrypted client's input data $\mathbf{c}=\text{Enc}(\text{Key},\mathbf{x})$. $r_2$ is used as the randomness input to $\Pi_{\threshold,N}$ for producing a random $\threshold$-out-of-$N$ share\footnote{Note that our implementation of $\threshold$-out-of-$N$ secret sharing produces random shares without any client’s identity.} $s_k \in F_q$ of $r_1$. $\Pi_{\threshold,N}$ operates in a finite field $F_q$ for some prime $q >0$ with information-theoretic security~\cite{bellare2020reimagining}. $\Pi_{\threshold,N}$ consists of two algorithms: i) share: producing a random share of the data with a particular randomness; ii) recover: reconstructing the data given at least its $\threshold$ valid shares. $r_3$ is used as a \emph{tag} informing the aggregation server which shares to combine to recover the encryption key. 
Each client constructs their message as {${\text{sbm} \leftarrow (\mathbf{c}, s, r_3)}$}.

Then, each client performs a Bernoulli test on whether to participate: with probability $p_s=n/N$ (where $n$ is the expected size of the sampled clients) it sends its encrypted message {${\text{sbm} \leftarrow (\mathbf{c}, s, r_3)}$} to the aggregation server, otherwise it abstains.

\subsection{Dummy Data Injection}
\label{sec:dummydataInject}
While the aggregation server cannot recover the data values submitted by less than \threshold clients, it does observe the tags of these unrevealed submissions.\footnote{The cost that we pay in favour of enabling the aggregation server to do the aggregation by itself without any interaction with other servers/clients in practical scenarios.} The aggregation server thus learns the multiplicity of unrevealed submissions sharing the same tag, which could potentially expose information about their underlying values if the server possesses additional side information. To control this leakage, \name adds dummy submissions such that the amount of information that the aggregation server can learn about these tags is bounded within the DP guaranteed range (Algorithm~\ref{alg:dummy}\footnote{The efficient version of Algorithm~\ref{alg:dummy-modified} used in our security proof (Appendix~\ref{app:secproof}). The only difference is that the client constructs dummy submissions locally instead of interacting with the randomness server, trivially tolerated by our security proof.}). Dummy data makes the histogram of unrevealed submissions differentially private.  This dummy data injection can be done by randomly selecting a client, and it does not affect the correctness and utility of the aggregations, as dummy data is automatically filtered out because each dummy group is smaller than the threshold (see line~3 in Algorithm~\ref{alg:dummydata}).

We use the truncated shifted discrete Laplace distribution supported on $\{0,...,2t\}$, denoted by $\text{TSDLap}(\lambda,t)$,  to generate a positive, bounded number of dummy data. 

\begin{definition}[Truncated Shifted Discrete Laplace Distribution] The Truncated Laplace Distribution on $\{0, ..., 2t\}$ is defined as:
\begin{equation}
f_{\text{TSDLap}(\lambda, t)}(c)=
\begin{cases}
    \frac{\exp{(-\frac{|c-t|}{\lambda})}}{A}, & \text{if $c \in \{0,...,2t\}$}\\
    0, & \text{otherwise,}
\end{cases}
\end{equation}
where $\lambda \in (0,1)$ is the scale parameter and $A=\sum_{c=0}^{2t} \exp{\left(-\frac{|c-t|}{\lambda}\right)}=1 + 2 \sum_{c=1}^{t} \exp\left(-\frac{c}{\lambda}\right)$ is the normalization constant.

\end{definition}

Section~\ref{sec:analysis} proves that subsampling clients, combined with dummy data and thresholding, provides DP guarantees~\cite{li2011provably,cormode2022sample}.

\begin{algorithm2e*}[t!]
\algsetup{linenosize=\tiny}
\small
\DontPrintSemicolon
\SetKwComment{Comment}{{\scriptsize$\triangleright$\ }}{$\quad\quad$}
\caption{\emph{Aggregation}: Aggregating and recovering client data}\label{alg:aggregation}
        \KwIn{Received $M$ clients' submissions $\{\text{sbm}_i\}_{i=1}^M$ where each submission ${\text{sbm}_i \leftarrow (\mathbf{c}, s, t)}$ contains an encrypted data $\mathbf{c}$, a random $\threshold$-out-of-$N$ share $s$ and a tag $t$, $\threshold$-out-of-$N$ secret-recovering scheme $\Pi_{\threshold,N}^{-1}$, \text{PseudorandomGenerator}}
        \KwOut{Decoded clients' data $\text{RecoveredData}$}
\BlankLine
\begin{algorithmic}[1]
\STATE $\text{RecoveredData}=\{\}$
\STATE $ \{\text{groups}\}= \text{GroupBasedOnTags}(\text{sbm}_i)$ \Comment*[r]{Grouping submissions based on their tags $t$ such that all submissions in each group share the same tag}
\FOR{$\text{group} \in \text{groups}$}
    \STATE $(r_1,\perp) \leftarrow \Pi_{\threshold,N}^{-1}(\{s_k\} \quad \forall{s_k \in  \text{group}})$ \Comment*[r]{Recovering the share if the group contains at least $\threshold$ submissions (i.e., secret-shares)}
    \IF{$r_1$}
    \STATE $ \text{Key} = \text{PseudorandomGenerator}(r_1)$  \Comment*[r]{Recovering the decryption key using $r_1$ to seed the pseudorandom generator}
    \STATE $\mathbf{x}=\text{Dec}(\text{Key},\mathbf{c})$ \Comment*[r]{Decrypting one of the data within the group}
    \STATE $\text{RecoveredData.append}(\underbrace{\mathbf{x}, \mathbf{x}, \dots, \mathbf{x}}_{|\text{group}| \text{ times}})$ \Comment*[r]{Outputting as many data as the size of the group}
    \ENDIF
    \ENDFOR
\STATE \textbf{Return} $\text{RecoveredData}$
\end{algorithmic}
\end{algorithm2e*}

\subsection{Data Aggregation and Recovery}

Clients submit their secret-shared values to the aggregation server through an anonymizing proxy, delinking the submitted value from any other information identifying the submitter (\eg{} IP address, etc). We deploy an Oblivious HTTP~\cite{I-D.ohai-http-oblivious} server which is an IETF draft standard. Oblivious HTTP removes client-identifying information from HTTP requests containing client submissions to blind the aggregator server from learning which client is submitting which reports, and which reports are being submitted by the same user. Following the literature,\footnote{\url{https://machinelearning.apple.com/research/learning-with-privacy-at-scale##AppleSecurity}} we assume that submissions do not contain timestamps of when data is generated. In practical deployments where timestamps are observed, we need to make the distribution of each timestamp independent of the messages and their source such that it gives no additional information about the sender. This could be done in various ways. For instance, whenever a client's data is encoded on their device, 
%1) the data gets temporarily stored on-device, and after a delay all clients simultaneously submit their data; or 2) 
the client locally draws a real number $\texttt{num}$ uniformly into $[0,1]$ and send the message at time $\texttt{num}$.

The aggregation server then recovers any values submitted by at least \threshold clients using the share recovery algorithm on the corresponding share values, $s_k$, to recover $r_1$ and thus the corresponding data value.
As described in Algorithm~\ref{alg:aggregation}, the aggregation groups submissions based on their tags such that all submissions in each
group share the same tag. Then, the aggregation can learn the submission within groups with cardinality of at least \threshold through performing the following sequential recoveries: 1) the share value $s$ from its \threshold secret
shares $s_k$; 2) $r_1$ from $s$; 3) the encryption key $\text{Key}$ from $r_1$; 5) the client submission using $\text{Key}$ as the decryption key.

\section{Privacy, Security, Utility and Communication Analysis}
\label{sec:analysis}

In this section, we analytically demonstrate that \name is a \emph{secure} protocol for producing \emph{private} and highly \emph{accurate} data outputs with \emph{low communication costs}.

\subsection{Privacy Analysis}

\begin{theorem}\label{theorem:DP}
   Consider $N$ clients generating a dataset $D=\{x_i\}_{i=1}^N$. Let $\varepsilon_{\text{Unre}}$ be the privacy budget used in the creation of dummy data (Algorithm~\ref{alg:dummy}). For $\varepsilon_{\text{Re}}>0$ and $\delta_{\text{Re}}\in(0,1)$, let {${p_s=\alpha (1-e^{-\varepsilon_{\text{Re}}})}$} and ${\threshold=\frac{1}{C_{\alpha}}\ln{(\frac{1}{\delta_{\text{Re}}})}}$ where $0 < \alpha \leq 1$ and $C_{\alpha}=\ln{(\frac{1}{\alpha})}-\frac{1}{1+\alpha}$. Then, the view of the aggregation in \name satisfies computational $(\varepsilon,\delta)$-DP  with $\varepsilon=\max(\varepsilon_{\text{Unre}}, \varepsilon_{\text{Re}})$ and $\delta=\max(\delta_{\text{Unre}}, \delta_{\text{Re}})$.
\end{theorem}
\begin{proof}

Let $D$ and $D'$ be two neighboring input datasets such that $D'$ is obtained by removing one client's data from $D$. We aim to show that the view of the aggregation server satisfies (computational) $(\varepsilon,\delta)$-DP. We split our analysis according to two mutually exclusive events (see Figure~\ref{fig:final-output}): either the value corresponding to the extra client in $D$ is revealed (i.e., the corresponding count is greater than or equal to $\tau$), or it is not.

\begin{figure}
    \centering
    \includegraphics[width=0.5\textwidth]{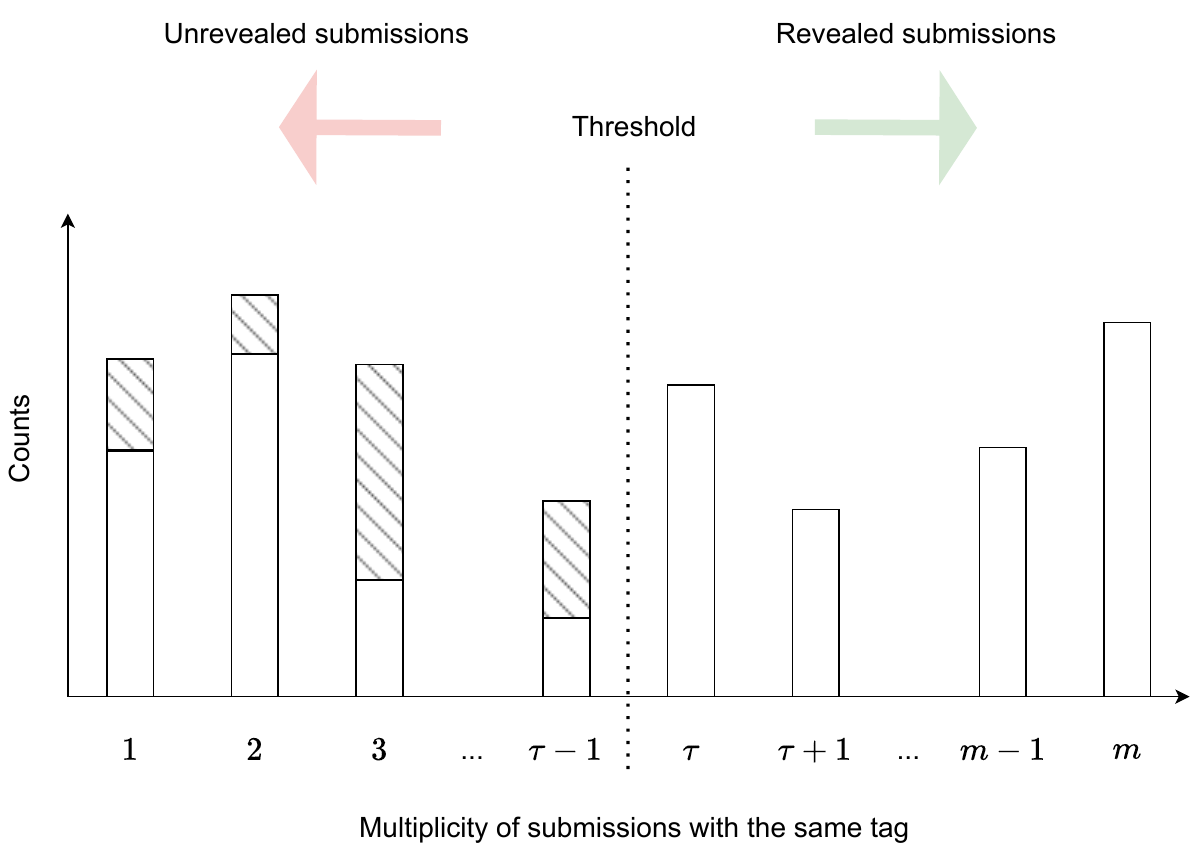}
    \caption{\name's output to the aggregation server consists of a histogram $\mathcal{H}$ of multiplicities where $\mathcal{H}_i$ represents the number of submissions with the same tag, with multiplicity $i$ and $i \in [m]$. This histogram is obtained based on submissions that each client sent with probability $p_s$ (empty bar) and dummy data (hatched bar).}
    \label{fig:final-output}
\end{figure}

\emph{Unrevealed submission.} For unrevealed submissions, the server only learns the multiplicity of submissions with the same tag. Let $\mathcal{Z}_{\text{Unre}}$ be the histogram of the unrevealed submissions computed as $\mathcal{Z}+\mathcal{N}$ where  $\mathcal{Z}$ is the histogram of multiplicities of the ``genuine'' submissions (i.e., $\mathcal{Z}_i$ counts the number of unrevealed tags with multiplicity $i$) and $\mathcal{N}$ is the noise corresponding to the addition of dummy contributions. Recall that dummy data are drawn from a domain disjoint from the original domain, so adding $i$ dummies with the same tag is equivalent to adding noise of value one to the $i$-th histogram entry $\mathcal{Z}_i$. As $\mathcal{Z}$ does not have any multiplicity above $\tau-1$, the protocol only adds $i \in [\tau-1]$ different such contributions. We know that noise $\mathcal{N}$ sampled from the truncated shifted discrete Laplace distribution $\text{TSDLap}(\lambda,t)$ on $\{0,...,2t\}$ with {$\lambda=\Delta/\varepsilon_{\text{Unre}}$} and $t=\Delta+\Delta/\varepsilon_{\text{Unre}}\log(2/\delta_{\text{Unre}})$ to ensure  $(\varepsilon_{\text{Unre}},\delta_{\text{Unre}})$-DP~\cite{bell2022distributed}.  We compute the sensitivity $\Delta$ as follows. Removing a client's data from $D$ decreases the count of the corresponding multiplicity $i$ by one while increasing the count of multiplicity $i-1$ by one, resulting in $\mathcal{Z}$ and $\mathcal{Z}'$ (computed on $D$ and $D'$ respectively) that differ in two adjacent entries $i$ and $i-1$:
\begin{equation}
\begin{cases}
   \mathcal{Z}_i &= \mathcal{Z}'_i + 1 \quad \text{entry $i$} \\   
   \mathcal{Z}_{i-1} &= \mathcal{Z}'_{i-1} - 1 \quad \text{entry $i-1$} \\
   \mathcal{Z}_{y} &= \mathcal{Z}'_{y} \quad \text{other entries $\forall y \notin \{i,i-1\}$}
\end{cases}   
\end{equation}
Therefore the sensitivity $\Delta=2$.

\emph{Revealed submission.} In the event where the differing submission is revealed, we can leverage DP guarantees of the sample-and-threshold approach~\cite{cormode2022sample}. For completeness and clarity, we give the full proof below. The bound on the ratio of the probability of the aggregation server receiving and decoding a group of submissions with the same tag and multiplicity $i$ on $D$ and $D'$ is computed as follows. Let $k$ be the multiplicity of the extra client's data item in $D$. The probability of seeing a count of $v \geq \tau$ copies of this item in the output of $D$ is given by the Binomial theorem:
\begin{equation}
    \label{eq:1}
    \binom{k}{v}(1-p_s)^{k-v}(p_s)^v,
\end{equation}
and the probability of seeing a count of $v$ copies of the same data in the output of $D'$ who holds $k-1$ copies of the data is 
\begin{equation}
    \label{eq:2}
    \binom{k-1}{v}(1-p_s)^{((k-1)-v)}(p_s)^v.
\end{equation}
Now, we can bound the ratio of probabilities of seeing data with a given count $v$ by dividing Eq.~\ref{eq:1} by Eq.~\ref{eq:2} which is $\frac{(1-p_s)k}{k-v}$.

Next, we show that the ratio $\frac{(1-p_s)k}{k-v}$ is between the interval $(e^{-\varepsilon_{\text{Re}}}, e^{\varepsilon_{{\text{Re}}}})$ except with some small probability. 

For the lower bound, we have
\begin{equation}
    \label{eq:3}
    e^{-\varepsilon_{\text{Re}}} \leq \frac{(1-p_s)k}{k-v},
\end{equation}
for any $v \geq 0$, which is satisfied if we ensure $p_s \leq 1-e^{-\varepsilon_{\text{Re}}} < 1$ (since $v = 0$ is the worst case).

For the upper bound, we have:

\begin{equation}
    \label{eq:4}
    \frac{(1-p_s)k}{(k - v)} \leq e^{\varepsilon_{\text{Re}}}.
\end{equation}

Rearranging the upper bound, we have:

\begin{equation}
    v \leq k(1 -e^{-\varepsilon_{\text{Re}}} + e^{-\varepsilon_{\text{Re}}} p_s)
\end{equation}

Note that:
\begin{enumerate}
    \item Since $p_s < 1$, then $p_s(1 - e^{-\varepsilon_{\mathcal{A}_{\text{Re}}}}) < 1 - e^{-\varepsilon_{\text{Re}}}$ and so $p_s < (1 - e^{-\varepsilon_{\text{Re}}} + e^{-\varepsilon_{\text{Re}}} p_s)$.
    \item The bound in \cref{eq:4} is greater than $kp_s$, the mean value.
    \item Since $p_s < 1$, then $(1 - e^{-\varepsilon_{\text{Re}}} + e ^{-\varepsilon_{\text{Re}}}p_s) < 1$, so we can define the probability $q = (1 - e^{-\varepsilon_{\text{Re}}} + e ^{-\varepsilon_{\text{Re}}}p_s)$.
\end{enumerate}

As all revealed submissions have a count at least equal to $\tau$, we can thus obtain $(\varepsilon_{\text{Re}}, \delta_{\text{Re}})$-DP by bounding the probability $\delta_{\text{Re}}$ of choosing a $v$ that is more than $\max(kq,\tau)$. Using the Chernoff-Hoeffding bound for the binomial distribution as done in~\cite{cormode2022sample}, we get $\delta_{\text{Re}} \leq \exp(-\frac{\tau}{q}D(q\|p))$. Therefore, sampling with probability $p_s=\alpha (1-e^{-\varepsilon_{\text{Re}}})$ and thresholding with ${\threshold=\frac{1}{C_{\alpha}}\ln{(\frac{1}{\delta_{\text{Re}}})}}$ where $0 < \alpha \leq 1$ and $C_{\alpha}=\ln{(\frac{1}{\alpha})}-\frac{1}{1+\alpha}$ provides $(\varepsilon_{\text{Re}},\delta_{\text{Re}})$-DP~\cite{cormode2022sample}. 
\end{proof}

Note that the aggregation server can publish the final histogram, which satisfies DP with the same $\varepsilon$ and $\delta$. Also, as discussed in Section~\ref{sec:problem_formulation} and Section~\ref{sec:method}, the randomness server learns nothing about the plaintext version of the client's data thanks to the underlying cryptographic schemes.

\subsection{Cryptographic Security}
Figure~\ref{functionality:nebula} represents the ideal functionality for \name. We leverage the methodology of~\cite{CCS:DSQGLH22} for producing consistent data encryption. Appendix~\ref{app:secproof} provides all correctness and security proofs of \name following similar arguments to~\cite{CCS:DSQGLH22}.

\begin{figure}
\begin{tcolorbox}[title=Ideal Functionality $\mathcal{F}_\text{Nebula}$]
\textbf{Participants:} 
\begin{itemize}[leftmargin=1.5em]
\item Aggregation server $S_A$ 
\item Randomness server $S_R$  
\item Clients $\{C_i\}_{i=1}^N$
\end{itemize}
\textbf{Public parameters:}
\begin{itemize}[leftmargin=1.5em]
    \item DP parameters $\varepsilon, \delta$ where $\varepsilon=\max(\varepsilon_{\text{Unre}}, \varepsilon_{\text{Re}})$ and $\delta=\max(\delta_{\text{Unre}}, \delta_{\text{Re}})$
    \item Noise parameters $\lambda=\frac{2}{\varepsilon_{\text{Unre}}}$ and $t=2+\frac{2}{\varepsilon_{\text{Unre}}} \log(\frac{2}{\delta_{\text{Unre}}})$.
    \item Threshold ${\threshold=\frac{1}{C_{\alpha}}\ln{(\frac{1}{\delta_{\text{Re}}})}}$ and Subsampling rate $p_s=\alpha (1-e^{-\varepsilon_{\text{Re}}})$
\end{itemize}

\textbf{Inputs:}
\begin{itemize}[leftmargin=1.5em]
        \item Client $C_i \in \{C_i\}^N_{i=1}$: provides input $(x_i, \text{aux}_i)$
        \item $S_R$: provides VOPRF keypair $(\text{msk}, \text{mpk})$
        \item $S_A$: provides nothing $\perp$ 
\end{itemize}
\textbf{Functionality:}
\begin{enumerate}
    \item For each client $C_i \in \{C_i\}^N_{i=1}$, sample $b_i \gets \text{Bern}(p_s)$
    \item $\mathcal{C} \gets \{C_i | b_i=1\}$ 
    \item Sample $\{\alpha_i\}^{\threshold - 1}_{i=1}$ where each $\alpha_i$ is sampled independently from $\alpha_i \leftarrow \text{TSDLap}(\lambda=2/\varepsilon_{\text{Unre}},t=2+2/\varepsilon_{\text{Unre}}\log(2/\delta_{\text{Unre}}))$. 
    \item For each $\alpha_i$, construct $\alpha_i$ groups of dummy clients of size $i$ with input $(\omega_{i,j}, \text{aux}_{i,j})$ for all $j \in \alpha_i$. Each $\omega_{i,j}$ is a distinct measurement outside the set of client measurements $\{x_i\}^N_{i=1}$. Call the set of all dummy clients $\mathcal{D}$.
    \item For each unique $x_\ell$ received from $\mathcal{C}\cup\mathcal{D}$, construct:
    \[
    \mathcal{E}_\ell = \left\{(x_\ell, \{x_j\}_{j \in J}, \threshold_\ell) : \left(J \subseteq [N]\right) \wedge \left(x_j = x_\ell\right) \right\}
    \]
    where $\threshold_\ell = \left|\{\text{x}_j\}_{j \in J}\right|$ is the number of sampled client measurements in $\mathcal{E}_\ell$.
    \item Let $\mathcal{Y}$ be an empty map.
    \item For each $\mathcal{E}_\ell$ where $\threshold_\ell \geq \threshold$, set $\mathcal{Y}[x_\ell] = \mathcal{E}_\ell$.
\end{enumerate}
\textbf{Outputs:}
\begin{itemize}[leftmargin=1.5em]
        \item $C_i$: learns nothing $\perp$ 
        \item $S_R$: learns what it would normally learn during the VOPRF exchange, $\{\mathcal{F}_{\text{VOPRF}}(\text{msk},x_i)\}_{i=1}^N$ 
        \item $S_A$: learns a $(\varepsilon, \delta)$ differentially private histogram of clients' data, $\mathcal{Y}$, and the cardinality of each group of inputs $n_\ell \equiv |S_\ell|$ where $S_\ell$ is the set of submissions from $\mathcal{C} \cup \mathcal{D}$ which share the same unique measurement $x_\ell$. 
\end{itemize}
\end{tcolorbox}
\caption{Ideal functionality for \name.}\label{functionality:nebula}
\end{figure}

\subsection{Communication Analysis}
\label{sec:comm-analysis}
Each client performs only one round of interaction with the randomness server to obtain the necessary randomness for the secret sharing. In particular, each client submits a 32 byte message consisting of their blinded hash value (see $b$ line 3 in Algorithm~\ref{alg:client-RandomnessServer}). In response, the client receives another 32 bytes (see $z$ line 5 in Algorithm~\ref{alg:client-RandomnessServer}) from the randomness server. Each client performs only one single interaction with the aggregation server. In particular, each client submits a 266 byte message, sbm (see line 6 in Algorithm~\ref{alg:client-encoding}), consisting of an alignment tag (32 bytes), a share of the encryption key (192 bytes), and their value (approximately 42 bytes). In addition to this, one client needs to send dummy data to the aggregation server.

\begin{prop}
    The expected and worst-case number of dummy data is $t \frac{(\tau - 1)\tau}{2}$ and $2t\frac{(\tau - 1)\tau}{2}$ where $t=2+2/\varepsilon_{\text{Unre}}\log(2/\delta_{\text{Unre}})$ is the expectation of the truncated shifted discrete Laplace distribution and $\tau$ is the threshold for pruning values.
\end{prop}
\begin{proof}
As discussed in Section~\ref{sec:method}, we use truncated shifted discrete Laplace distribution, $\text{TSDLap}(\lambda,t)$ on $\{0,...,2t\}$ to generate dummy data. The expectation of $\text{TSDLap}(\lambda,t)$ is $t$ and its maximum value is $2t$. We sample $\tau-1$ times from the truncated shifted discrete Laplace distribution and each time generate a group of submissions whose cardinality is the same as the bin value. Therefore, the expected and the maximum number of dummy submissions are $t \frac{(\tau - 1)\tau}{2}$ and $2t\frac{(\tau - 1)\tau}{2}$, respectively.
\end{proof}

As alignment tags for the dummy data are random they can be generated locally without any communication with the randomness server. However, the submitting client needs to send as many messages as the size of the group to the aggregation server.

\subsection{Utility Analysis}
The utility of \name is unaffected by the inclusion of dummy data: only the sampling rate $p_s$ and the threshold $\tau$ affect the accuracy of the histogram estimated by \name relative to the true histogram. Leveraging utility guarantees of sample-and-threshold approach~\cite{cormode2022sample}, we can show that \name reveals to the aggregation server, with high probability, any value whose frequency is sufficiently above the threshold.  

\begin{lemma}\label{th:utility}
    \name removes a value that is shared by $W$ clients with probability at most $\exp(-(p_sW-\threshold)^2\frac{1}{2Wp_s})$, where $p_s$ and $\threshold$ are \name's parameters: the client sampling rate and pruning threshold.
\end{lemma}

\section{Nested-\name: a variant for high-dimensional marginal histograms}

\begin{algorithm2e}[t!]
\algsetup{linenosize=\tiny}
\small
\DontPrintSemicolon
\SetKwComment{Comment}{{\scriptsize$\triangleright$\ }}{$\quad\quad$}
\caption{Nested-\name}\label{alg:client-nested-encoding}
        \KwIn{$N$ clients, each client $i$ holding a multi-dimensional data point $\mathbf{x}_i=[x_i^{(1)},x_i^{(2)},...,x_i^{(L)}]$ with $L$ attributes, all other inputs to Algorithm~\ref{alg:mechanism1}}
        \KwOut{Clients' multi-dimensional submissions revealed to the aggregation server}
\BlankLine
\begin{algorithmic}[1]
\STATE $(\varepsilon_{\text{Re}},\delta_{\text{Re}}),(\varepsilon_{\text{Unre}},\delta_{\text{Unre}}) \leftarrow (\varepsilon,\delta)$ and $p_s,\threshold \leftarrow (\varepsilon_{\text{Re}},\delta_{\text{Re}})$ 
\FOR{$i  = 1,\ldots, N$}
\STATE $\text{SBM}=\{\}$
\FOR{$\ell = 1,\ldots, L$}
\STATE $\mathbf{x}_i^{(:\ell)}=[x_i^{(1)},\cdots,x_i^{(\ell)}]$
\STATE $r_i^{(\ell)}=\textit{Client-RandomnessServer}(\mathbf{x}_i^{(:\ell)},\text{pp},H(\cdot))$ 
\STATE $ \text{sbm}_i^{(\ell)}, \text{Key}_i^{(\ell)} \leftarrow \textit{LocalSecretSharing}(\mathbf{x}_i^{(:\ell)},r_i^{(\ell)},\Pi_{\threshold,N})$ 
\IF{$\ell == 1$}
\STATE $\widehat{\text{sbm}}_i^{(\ell)} \leftarrow \text{sbm}_i^{(\ell)}$
\ELSE
\STATE $\widehat{\text{sbm}}_i^{(\ell)} \leftarrow {\text{Enc}}({\text{Key}_i^{(\ell-1)},\text{sbm}_i^{(\ell)}})$ 
\ENDIF
\STATE $\text{SBM}.\text{append}(\widehat{\text{sbm}}_i^{(\ell)})$
\ENDFOR
\STATE $z_i=\text{Random}([0,1])$ 
\IF{$z_i \leq p_s$}
\STATE Submit SBM to the aggregation server
\ENDIF
\ENDFOR
\STATE $\text{Dummy}=\textit{DummyDataCreation}(\threshold,\text{TSDLap}(\cdot),(\varepsilon_{\text{Unre}},\delta_{\text{Unre}}))$ 
\STATE Submit Dummy to the aggregation server
\STATE $\text{ReceivedData}=(\text{SBM}\cup\text{Dummy})$ 
\FOR{$\ell = 1,\ldots, L$}
\IF{$\ell == 1$}
\STATE $\text{RecoveredData}^{(\ell)} \leftarrow \textit{Aggregation}(\text{ReceivedData}^{(\ell)})$
\ELSE
\STATE $\text{ReceivedDataWithKeys}^{(\ell)} \leftarrow {\text{Dec}}({\text{Key}^{(\ell-1)},\text{ReceivedData}^{(\ell)}})$ 
\STATE $\text{RecoveredData}^{(\ell)}=\textit{Aggregation}(\text{ReceivedDataWithKeys}^{(\ell)})$ 
\ENDIF
\ENDFOR
\STATE \textbf{Return} $\text{RecoveredData}$
\end{algorithmic}
\end{algorithm2e}

\looseness=-1\sloppy In some scenarios, clients' data consist of multiple attributes~\cite{zhang2018calm,wang2019answering,leith2021mobile}. In this case, each client $i$ holds a multi-dimensional data point that is a vector of $L \geq 2$ attributes of the form $\mathbf{x}_i=[x_i^{(1)},x_i^{(2)},...,x_i^{(L)}]$. Consider the following motivating and practical example in the case of telemetry~\cite{CCS:ErlPihKor14}. Clients generate some five-dimensional crash reports while using a Web Browser. Clients are anonymous and each dimension is a separate attribute: \textit{i)} URL visited; \textit{ii)} the underlying operating system; \textit{iii)} the state of the device's battery; \textit{iv)} session; and \textit{v)} token IDs~\cite{doi.org/10.48550/arxiv.1808.01718}. These attributes form a client's crash report. The service provider would like to learn the \emph{marginal histogram} (\ie the frequency among any joint sequence of attributes) to optimize and improve their application. Therefore, the service provider wants to maximize the utility of marginal histogram estimations and get better utility than treating client data that are made up of multiple attributes as a single data. Indeed, treating all attributes as a single data point means only those clients whose multi-dimensional data match exactly across all attributes are considered to have the same value, and this typically rarely happens in real datasets with many attributes. One straightforward solution would be to share each individual attribute or sequence of joint attributes, but this increases privacy risks.

To address these limitations, we propose Nested \name (Algorithm~\ref{alg:client-nested-encoding}) that employs a novel multi-dimensional data encoding  with a ``hierarchy-of-priority'' in which each client constructs a ciphertext by iteratively encrypting their ordered attributes such that the decrypting process halts when facing a low-frequency attribute that might risk the privacy of clients.

\looseness=-1 \myparagraph{Multi-dimensional local data encryption.} In Nested-\name, each client $i$ creates $L$ sequential prefixes $\mathbf{x}_i^{(:1)},\dots,\mathbf{x}_i^{(:L)}$ such that each prefix $\mathbf{x}_i^{(:\ell)}=[x_i^{(1)},\cdots,x_i^{(\ell)}]$ contains the sequence of attributes from the beginning to its corresponding index $\ell\in \{1,\dots,L\}$. These prefixes enable to capture joint histograms of multiple attributes instead of each individual attribute. Each client encodes prefixes $\mathbf{x}^{(:\ell)}$ such that rare prefixes (\ie a sequence of attributes which are not common across clients) cannot be decoded (\ie kept hidden from the server). Each client $i$ secret-shares each prefix $\ell$ through running Algorithm~\ref{alg:client-RandomnessServer} and Algorithm~\ref{alg:client-encoding} and construct message {$\text{sbm}_i^{(\ell)}$}. Submitting $\text{sbm}_i^{(1)}, \text{sbm}_i^{(2)}, ..., \text{sbm}_i^{(\ell)}$ separately in $\ell$ individual messages would result in two issues: \textit{i)} it increases privacy loss and reveals all tags that might leak information; and \textit{2)} it increases the overhead for both clients and servers. 
To address these issues, we chain the prefix contributions of each client (see lines 6-9 of Algorithm~\ref{alg:client-nested-encoding}) and create one single super-message $\textsf{SBM}$ such that decoding the attributes in the previous prefix would only then allow unlocking the next-longer prefix. In particular, we construct super-messages that can be iteratively opened to reveal higher levels of granularity (more attributes) when the previous prefix is shared by at least $\threshold$ clients. To do this, client $i$ encrypts each prefix $\text{sbm}_i^{(\ell)}$ with the key of its previous prefix which gets revealed once the previous prefix is decoded. Each prefix at layer \(\ell\) (except for the first) is then encrypted with the key of the previously encoded prefix. In particular,  each client computes an encrypted ciphertext $\widehat{\text{sbm}}_i^{(\ell)} \leftarrow {\text{Enc}}({\text{Key}^{(\ell-1)},\text{sbm}_i^{(\ell)}})$ with the described symmetric encryption operation for each \(\ell \in \{2,\dots,L\}\). Finally, each client $i$ creates the super-message as the tuple $\textsf{SBM}=(\widehat{\text{sbm}}_i^{(1)},\ldots,\widehat{\text{sbm}}_i^{(L)})$ and sends it to the server based on the outcome of Bernoulli test discussed in Section~\ref{sec:method}.

We assume that attributes come with a natural order (i.e., hierarchy-of-priority). However, this ordering affects the utility of Nested-\name as attributes towards the end of the clients’
submissions are less likely to be learned by the aggregation server. An interesting future direction is
to optimize the ordering of client attributes based on domain
knowledge (e.g. the distribution of data itself) to achieve
high utility.

We also note that the utility improvement of this multi-dimensional encoding might come at a cost in terms of privacy. An aggregation server with perfect background knowledge (full knowledge of all records in $D$, and full knowledge of the victim record), aiming to infer whether the victim record is in the input dataset or not, can recover some information. In particular, the aggregation server can observe tags of unrevealed prefixes $\widehat{\text{sbm}}_i^{(\ell)}$ whose immediate preceding prefixes $\widehat{\text{sbm}}_i^{(\ell-1)}$ are decoded. However, it is common to ignore this leakage in practice~\cite{desfontaines2019sok} as the above privacy leakage happens for pathological datasets with extremely skewed distribution and mostly binary values. Several ways have been proposed to address this leakage in the literature, such as relaxing the definition of differential privacy by considering practical data distributions~\cite{bassily2013coupled,duan2009privacy,kifer2012rigorous,bhaskar2011noiseless,desfontaines2019sok,li2011provably}.

\section{Experiments}
\label{sec:experiment}

We implement \name (\url{https://github.com/brave-experiments/Nebula-CCS2025}) and empirically validate: \textbf{i) Effectiveness in estimating accurate but private histograms}: in complement to the analytical privacy and utility guarantees of Section~\ref{sec:analysis}, we empirically demonstrate that the histogram estimated by \name is close to the true histogram constructed from all clients' original data while ensuring strong privacy guarantees. \textbf{ii) Efficiency in private and secure data collection}: In complement to the analytical communication costs of Section~\ref{sec:analysis}, we empirically demonstrate the ability of \name to scale to real-world use cases thanks to its low computational, bandwidth and financial costs.

%\subsection{Datasets and Setup}

We assess the performance of \name on real-world three datasets. Two of these are \emph{privacy-sensitive in nature}---the IPUMS Census dataset and the Foursquare dataset~\cite{Yang2016a}. We also assess the performance of \name on the Complete Works of Shakespeare as it is commonly used in histogram estimation literature~\cite{cormode2022sample}.

\myparagraph{IPUMS}. We use the Integrated Public Use Microdata Series of United States census data ({\url{https://usa.ipums.org/usa/}). We consider 15,537,785 data points representing persons through 5 attributes: SEX, marriage status (MARST), RACE, education (EDUC), AGE.

\myparagraph{Foursquare dataset} is derived from the mobile app ``Foursquare City Guide'', which takes advantage of a user's location to guide them to highly-rated places like restaurants and bars, while a social networking feature lets the user's friends know what places they visit. The dataset contains 33,263,633 check-in events at 3,680,126 venues (in 415 cities in 77 countries).  Each venue in the dataset (e.g. a restaurant) comes with a latitude and longitude granular enough to identify it uniquely.
We pre-process the dataset by extracting the country code and latitude/longitude pairs of each check-in event.
The result is a CSV file of 33,263,633 rows where each line contains the location information for one venue visit by one client.

\myparagraph{Shakespeare dataset.} We also consider the complete works of William Shakespeare,\footnote{Plain text edition from \url{https://cs.stanford.edu/people/karpathy/char-rnn/shakespeare_input.txt}}
as if clients were each contributing an individual word to generate a frequency
distribution. We split the text on whitespace, and apply basic normalization of
punctuation and capitalization. This results in a sequence of 832,301 values out
of a set of 29,257 unique words. Note that the frequency distribution is highly
peaked, in part because no stop words were removed.
To study the effect of domain size, we also sort words into bins based on the lower $b$ bits of their SHA256 hash, and apply the same algorithms to the bin index, creating a deterministic mapping consistent with what clients could perform before submission.

%\subsection{Setup}

\myparagraph{Evaluation metrics.} We evaluate the effectiveness and efficiency of \name as follows. \textbf{Effectiveness}: We measure utility as a (scalar) error that quantifies the difference between the estimated and the true histograms: i) we represent each histogram as a vector, where the length of the vector corresponds to the total number of bins, and each entry represents the count in the corresponding bin; ii) we normalize each histogram by dividing each bin counts by the total counts; iii) we compute the $l_1$ norm of the difference between the estimated and original normalized histograms. \textbf{Efficiency}: We evaluate the various costs of our framework through (1) Computational costs measured as CPU running time for both the client-side encoding step and the server-side aggregation step; (2) Financial costs based on those running times and per-CPU-hour server rental prices,  and (3) Bandwidth costs by measuring the size of a submission to the aggregation and randomness servers.

\myparagraph{Parameters.} As discussed in Section~\ref{sec:analysis}, \name's privacy budget is computed as $\varepsilon=\max(\varepsilon_{\text{Unre}}, \varepsilon_{\text{Re}})$ and $\delta=\max(\delta_{\text{Unre}}, \delta_{\text{Re}})$. We set the parameters of \name---threshold \threshold, sampling rate $p_s$ and shift $t$ in TDSLap---such that we obtain a desired $(\varepsilon,\delta)$ privacy guarantees and a trade-off between utility and communication costs as follows: (1) Set $\varepsilon_{\text{Re}} = \varepsilon \leq 1$ (smaller $\varepsilon$ provides a stronger privacy guarantee) and  $0 < \alpha \leq 1$ to a desired privacy budget and a constant, respectively, and compute the sampling rate as $p_s=\alpha (1-e^{-\varepsilon})$; (2) Set $\delta_{\text{Re}}=\delta$ to be very small (less than the reciprocal of the total number of clients) and compute the threshold as $\threshold=\frac{1}{C_{\alpha}}\ln{(\frac{1}{\delta})}$ where $C_{\alpha}=\ln{(\frac{1}{\alpha})}-\frac{1}{1+\alpha}$; and (3) set $t=2+2/\varepsilon_{\text{Unre}}\log(2/\delta_{\text{Unre}}))$ such that $\varepsilon_{\text{Unre}} \leq \varepsilon_{\text{Re}}$ and $\delta_{\text{Unre}} \leq \delta_{\text{Re}}$. In particular, setting $\varepsilon=1$, $\alpha=1/6$ and $\delta=10^{-8}$ yields $p_s= 0.105$, $\threshold=20$ and $t=15$. 

\myparagraph{Implementation of cryptographic primitives.} 
We use the Adept Secret Sharing framework, implemented in v0.2.3 of the “adds” crate. For randomness stretching / OPRF, we use the Puncturable Partially Oblivious Pseudorandom Function algorithm, implemented with v0.4.1 of the “pporpf” crate. Finally, for Symmetric cryptography and hash functions, we pull from the Strobe protocol framework, implemented in v0.10.0 of the ``strobe-rs'' crate.

\subsection{Utility Comparison to Existing Works}

\begin{table}[t]
    \centering
    \caption{Comparison of the utility of \name against: i) DP approaches--including all local, shuffle and central models-- under a fixed DP guarantee of $\varepsilon=1$; and ii) the threshold-aggregation technique STAR implementing a $K$-anonymity privacy protection ($K=20$).  Utility is assessed as the error between the estimated and true histograms ($\downarrow$: the lower, the better) on Shakespeare (Shak.), IPUMS and Foursquare (Fours.) datasets. \name enforces differential privacy (unlike STAR) and achieves lower error than both local and shuffle DP methods. It brings utility closer to that of central DP, without requiring trust in a central server.}
    \begin{tabular}{|l|c|ccc|}
    \Xhline{3\arrayrulewidth} 
\multirow{2}{*}{Approaches} & \multirow{2}{*}{DP} & \multicolumn{3}{c|}{Utility ($\downarrow$)} \\
    &   & Shak. & IPUMS & Fours. \\
\Xhline{3\arrayrulewidth} 
      Secure threshold-aggregation (STAR~\cite{CCS:DSQGLH22}) &  \xmark & 0.1934 & 0.0314 & 0.6217\\
      \hline
      Central (Laplace noise~\cite{dwork2014algorithmic}) & \cmark & 0.0348 & 0.0044 & 0.0005\\
      Shuffle (Multi-message Bernoulli noise~\cite{balcer2019separating})        & \cmark & 1.2910 & 0.4344 & 1.4184\\
      Local  (General-Randomized Response~\cite{KairouzOV14})    & \cmark & 1.5921 & 1.6622 & 2.0000\\
      \name    & \cmark &  0.5772 & 0.1932 & 1.3999\\
      \Xhline{3\arrayrulewidth} 
    \end{tabular}
    \label{tab:utility-comparision}
\end{table}

We evaluate our DP data collection framework, \name, against existing DP methods, and (non-DP) threshold-aggregation systems. Specifically, we compare \name to four baselines: \textbf{Local DP based on the generalized randomized response mechanism}~\cite{KairouzOV14}: Dividing the total privacy budget $\varepsilon$ evenly across all attributes, each attribute is locally perturbed by the client using randomized response under the corresponding per-attribute budget, and the final report aggregates these noisy responses. \textbf{Shuffle DP via multi-message Bernoulli noise addition}~\cite{balcer2019separating}: Each client reports a value 1 for their true bin value, and independently reports a value 1 for each bin of the histogram with probability $p_{\text{Shuffle}}=1- \frac{50}{\varepsilon^2N}\ln{(2/\delta)}$. A shuffler collects the messages, concatenates and randomly permutes them, and forwards the shuffled message to the aggregation server. The aggregator server can then recover an unbiased estimate of the histogram. \textbf{Central DP via Laplace noise addition~\cite{dwork2014algorithmic}}: the aggregation server collects the raw client's data, computes the true histogram and adds Laplace noise to each bin). \textbf{STAR}~\cite{CCS:DSQGLH22} with $K$-anonymity protection (no DP): Each client encodes their value as $K$-out-of-$N$ secret shares and sends to the aggregation server (without Poisson sampling and dummy data that \name introduces for achieving DP guarantees). The aggregation server decrypts values submitted by at least $K$ clients.

\begin{figure}[t!]
\centering
\begin{tikzpicture}
\begin{axis}[cycle list name=color list,
                  small,
                  axis lines=left, 
                  width=6cm,
                  height=5cm,
                  %ymode=log,
ymajorgrids,
                  ymin=0.00,
                  ymax=1,
enlarge x limits=0.04,
                  ylabel={\small Absolute error},
                  xlabel={\small Number of histogram bins},
                  scaled ticks=false,
xtick={6,7,...,14},
                  xticklabels={$2^{6}$,$2^{7}$,$2^{8}$,$2^{9}$,$2^{10}$,$2^{11}$,$2^{12}$,$2^{13}$,$2^{14}$},
legend style={at={(1.4,0.95)}, anchor=north},
                  ]

\addplot+[green,line width=2pt, error bars/.cd,y fixed,y dir=both, y explicit]
              table[x=bits,y=nebula,y error=nebula-sd] {shakespeare-bits-wShuffle.txt};
\addplot+[red,dotted, line width=2pt, error bars/.cd,y fixed,y dir=both, y explicit]
              table[x=bits,y=ldp,y error=ldp-sd] {shakespeare-bits-wShuffle.txt}; 
\addplot+[blue,densely dashdotted,line width=2pt, error bars/.cd,y fixed,y dir=both, y explicit]
              table[x=bits,y=gdp,y error=gdp-sd] {shakespeare-bits-wShuffle.txt};
\addplot+[black,dashdotted,line width=2pt, error bars/.cd,y fixed,y dir=both, y explicit]
              table[x=bits,y=shuffle,y error=shuffle-sd] {shakespeare-bits-wShuffle.txt};              
\legend{\small \name, \small Local DP, \small Central DP, \small Shuffle DP}
\end{axis}
\end{tikzpicture}
\caption{Utility of \name compared with local, Shuffle and central differential privacy applied to the Shakespeare database as a function of histogram bins using an $\varepsilon=1$ DP privacy guarantee. The word-frequency estimate of \name is more accurate than local and shuffle DP while while removing the trust of the central DP models on the server.}\label{fig:shakespeare-bins}
\end{figure}

Table~\ref{tab:utility-comparision} shows the histogram estimation error of \name, local DP, Shuffle DP, central DP and STAR on the three datasets. Results demonstrate that \name is more effective than the alternative local DP and shuffle DP approaches in the collection of high-utility data with strong DP guarantees: the utility of \name is closer to the utility of the central model of DP in which clients must trust the server. 
The absolute error in Shakespeare is the lowest. This is because Shakespeare's domain size is smaller than those of the other two datasets, IPUMS and Foursquare. In these datasets, a value represents a combination of multiple attributes, causing the domain size to grow exponentially with the number of attributes.

We further analyze the effect of the domain size in Figure~\ref{fig:shakespeare-bins}. Words from the Shakespeare dataset are mapped by hash value into between 64 and 16384 bins. The smaller the number of bins, the lower the absolute error. As the domain size shrinks, the error trends toward that of the central DP method, consistent with our above explanation of the performance difference between datasets. Across all domain sizes, \name consistently achieves significantly lower absolute error in the estimated histogram compared to local DP. As the domain size approaches the true domain of the dataset, the utility advantage of \name over Shuffle DP becomes increasingly evident. Additionally, \name is highly communication-efficient: each client transmits information solely about its held item, while Shuffle DP produces a message for each possible value. The domain-independent communication complexity of \name is especially beneficial in large domains~\cite{cormode2022sample}.

\looseness=-1 Conversely, the absolute error in Foursquare is very high: this is because Foursquare consists of multi-attribute data, resulting in a large domain consisting of specific geographic coordinates. Next, we discuss a variant of \name that can decrease the error in multi-attribute cases such as  Foursquare and IPUMS census.

\subsection{Utility Improvements via Nested-\name}

We analyze the impact of multi-dimensional encoding on the utility of \name in estimating the marginal histogram across both multi-dimensional dataset: the IPUMS dataset and the Foursquare dataset. We compare the absolute error of \name against Nested \name when estimating the histogram for each prefix.

Figure~\ref{fig:AbsErrIPUMS} (first row) demonstrates that multi-dimensional encodings improve the ability of \name to collect high-utility multi-dimensional IPUMS data. As the number of attributes increases in a prefix, the absolute error of the histogram estimation increases (when including all attributes, we recover the results of Table~\ref{tab:utility-comparision}). This is because increasing the number of attributes in a prefix decreases the chance of having more copies of items, amplifying the costs of sampling and pruning on revealing the item at the output to the server: the chance of sampling a low-frequency item decreases and it is more likely the items will be pruned (see Lemma~\ref{th:utility}).

\looseness=-1 We now turn to the Foursquare dataset where each element consists of geographic coordinates and a country code, which are not independent attributes.
However, the chained prefix encoding can still be applied to improve utility by coarse-graining the venue locations. If each visit is split into the country code and successive digits of the coordinates, a sequence of 8 attributes is produced reporting the event location with increasing granularity. With this encoding, partial recovery of joint attributes amounts to recovering regional aggregate popularity at multiple scales. Figure~\ref{fig:AbsErrIPUMS} (second row) shows the improvement in the absolute error using this multi-dimensional nested encoding for the Foursquare dataset. To further demonstrate the effectiveness of \name and our multi-dimensional nested encoding, we compare the estimated Nested \name histogram and the true Foursquare histogram of country codes in Figure~\ref{fig:hist}. We observe that Nested \name preserves the relative frequencies across attribute values. For example, the most popular items stay popular in the estimated histogram.

\begin{figure}[t!]
\centering
\small
\begin{tikzpicture}
\begin{axis}[cycle list name=color list,
axis lines=left, 
                  width=8cm,
                  height=4cm,
ybar,
                  ymin=0.00,
                  ymax=0.20,
                  %title= IPUMS USA dataset,
                  enlarge x limits=0.1,
                  ylabel={Absolute error},
                  xlabel={\small IPUMS USA Prefix (\ie a sequence of attributes)},
symbolic x coords={S, SM, SMR, SMRE, SMREA},
                  xtick=data,
                  x axis line style={draw=none},
                  legend style={at={(.25,0.95)}, anchor=north},
]

\addplot[fill=black!60!green,draw=black!60!green] coordinates {(S,0.0006) (SM,0.0019) (SMR,0.0041) (SMRE,0.0118) (SMREA,0.1932)};\label{ss-ipums}

\addplot[fill=green,draw=green] coordinates {(S,0.0032) (SM,0.0671) (SMR,0.1170) (SMRE,0.1410) (SMREA,0.1932)};\label{vanilla-ss-ipums}

\legend{\small Nested \name, \small \name}

\end{axis}
\end{tikzpicture}
\begin{tikzpicture}
\begin{axis}[cycle list name=color list,
axis lines=left, 
                  width=9.1cm,
                  height=4cm,
ybar,
                  ymin=0.00,
                  ymax=1.65,
                  legend style={at={(.25,0.95)}, anchor=north},
                  enlarge x limits=0.1,
                  ylabel={Absolute error},
                  xlabel={\small Foursquare Prefix (\ie a sequence of attributes)},
                  ytick={0.0,0.4,0.8,1.2, 1.6},
                  yticklabels={0.0,0.4,0.8,1.2,1.6},
                  symbolic x coords={$\mathbf{x}^{(1)}$, $\mathbf{x}^{(2)}$, $\mathbf{x}^{(3)}$, $\mathbf{x}^{(4)}$, $\mathbf{x}^{(5)}$, $\mathbf{x}^{(6)}$, $\mathbf{x}^{(7)}$, $\mathbf{x}^{(8)}$},
                  xtick=data,
                  x axis line style={draw=none},
]

\addplot[fill=black!60!green,draw=black!60!green] coordinates {($\mathbf{x}^{(1)}$,0.0025) ($\mathbf{x}^{(2)}$,0.0069) ({$\mathbf{x}^{(3)}$},0.0191) ({$\mathbf{x}^{(4)}$},0.2535) ({$\mathbf{x}^{(5)}$},1.0887) ({$\mathbf{x}^{(6)}$},1.3822) ({$\mathbf{x}^{(7)}$},1.3999) ({$\mathbf{x}^{(8)}$},1.3999)};\label{ss-foursquare}

\addplot[fill=green,draw=green] coordinates {($\mathbf{x}^{(1)}$,0.3904) ($\mathbf{x}^{(2)}$,0.4636) ({$\mathbf{x}^{(3)}$},0.5113) ({$\mathbf{x}^{(4)}$},0.7827) ({$\mathbf{x}^{(5)}$},1.2664) ({$\mathbf{x}^{(6)}$},1.3926) ({$\mathbf{x}^{(7)}$},1.3997) ({$\mathbf{x}^{(8)}$},1.3999)};\label{vanilla-ss-foursquare}

\legend{\small Nested \name, \small \name}

\end{axis}
\end{tikzpicture}
\caption{Improving the utility of \name in estimating the histogram on the multi-attribute IPUMS dataset and Foursquare dataset using multi-dimensional data encoding, Nested \name (Algorithm~\ref{alg:client-nested-encoding}). IPUMS contains 5 attributes--S: Sex; M: Marriage status; R: Race; E: Education; A: Age. Foursquare dataset contains
8 prefixes: $\mathbf{x}^{(1)}=[x_1]$, $\mathbf{x}^{(2)}=[x_1,x_2]$, $\cdots$, $\mathbf{x}^{(8)}=[x_1,\cdots,x_{8}]$. We compute the utility as the absolute error between the original histogram and the estimated histogram. Multi-dimensional data encoding significantly improves the absolute error of each marginal histogram (\ie histogram of each sequence of joint attributes).}
\label{fig:AbsErrIPUMS}
\end{figure}

\begin{figure*}[t]
    \centering
    \includegraphics[width=0.9\textwidth]{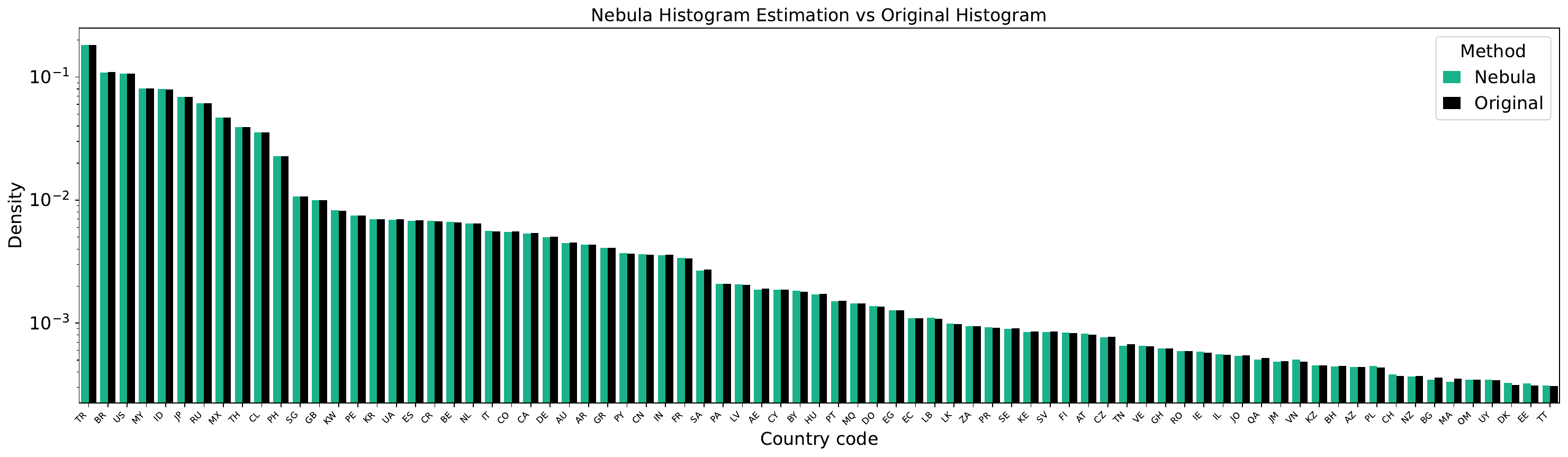}
    \caption{Original and estimated histogram obtained privately by \name using Foursquare dataset. The private histogram estimated by \name is close to the histogram of the original data. \name also preserves the relative order across values.}
    \label{fig:hist}
\end{figure*}

\subsection{\name is Efficient}
\label{sec:expEffic}

We evaluate computational, bandwidth, and financial costs.

\myparagraph{Computational costs.} We measure the CPU running time of our framework on AWS r6a machine for each client, the aggregation server and the randomness server, separately. Each client performs two sets of computations: (1) obtain randomness by interacting with the randomness server (Algorithm~\ref{alg:client-RandomnessServer}); and (2) encode attributes to be submitted to the aggregation server (Algorithm~\ref{alg:client-encoding}). As shown in Table~\ref{tab:runningtime-datasets-as}, these steps (1) and (2) take 4.96 and 0.85 milliseconds, respectively, for each submission from the Foursquare dataset with the chained-prefix encoding and 8 attributes. Running times are even lower for the IPUMS dataset as each client holds fewer attributes (five). This running time scales proportionately for the Shakespeare dataset, where each client submission consists of a single word without prefix encoding. Per-submission running time is comparable for all three datasets when the whole value is encoded as a single attribute. The running time of the randomness server (last row in Table~\ref{tab:runningtime-datasets-as}) is very low; as it only needs to perform one OPRF evaluation on each client's request using its secret key. This takes only 0.48 milliseconds for each Foursquare client submission. For simplicity, we omit the zero-knowledge proof steps of the verifiable OPRF in these benchmarks. Table~\ref{tab:runningtime-datasets-as} shows the CPU running time of our framework for the aggregation server. It takes 345 seconds for the server to process all 33,263,633 client submissions from the Foursquare dataset. Therefore, \textbf{\name introduces only a very little computational overhead for all parties---clients, the randomness server and the aggregation server}.

\begin{table}[t!]
    \centering
    \small
    \caption{Efficiency of \name in terms of running time on Foursquare and IPUMS datasets, in \emph{milliseconds per submission}. \textbf{The computational overhead of \name for clients and the randomness server is very small.}}
    \begin{tabular}{|l|l|lll|}
    \Xhline{3\arrayrulewidth} 
\multirow{2}{*}{Party} & \multirow{2}{*}{Function} & \multicolumn{3}{c|}{Dataset} \\
      &  & Foursquare  & IPUMS & Shakespeare \\
\Xhline{3\arrayrulewidth} 
      \multirow{2}{*}{Client}  & Encode & $4.96$ & $3.07$ & $0.42$  \\
      & Randomness   &  $0.85$  & $0.53$ & $0.21$ \\
 \hline
      Server & OPRF evaluation &  $0.48 $  & $0.30$ & $0.06$ \\
\Xhline{3\arrayrulewidth} 
    \end{tabular}
    \label{tab:runningtime-datasets-as}
\end{table}
\begin{table}[t!]
    \centering
    \small
    \caption{Efficiency of \name in terms of running time on three datasets, in \emph{seconds for all submissions}. \textbf{The computational overhead of \name for the aggregation server is small.}}
    \begin{tabular}{|l|l|lll|}
    \Xhline{3\arrayrulewidth} 
\multirow{2}{*}{Party} & \multirow{2}{*}{Function} & \multicolumn{3}{c|}{Dataset} \\
 & & Foursquare  & IPUMS & Shakespeare \\
\Xhline{3\arrayrulewidth} 

Server & Decode & $345$ & $101$ & $8$ \\
      \Xhline{3\arrayrulewidth} 
    \end{tabular}
    \label{tab:runningtime-datasets}
\end{table}
\begin{table}[t!]
    \centering
    \caption{Efficiency of \name in terms of bandwidth costs on Foursquare and IPUMS datasets in bytes per submission. \textbf{The bandwidth cost of \name is very small}.}
    \begin{tabular}{|l|lll|}
    \Xhline{3\arrayrulewidth} 
Interaction & Foursquare  & IPUMS & Shakespeare \\
\Xhline{3\arrayrulewidth} 
Client\&RandomnessServer & 256 & 160 & 32 \\
Client\&AggregationServer & 2465 & 1470 & 373 \\
      \Xhline{3\arrayrulewidth} 
    \end{tabular}
    \label{tab:bandwidth-datasets}
\end{table}

\myparagraph{Bandwidth costs.} Table~\ref{tab:bandwidth-datasets} shows the bandwidth costs that \name introduces for each client. The total communication costs of running \name for each client is at most 2.7 KB (not including framing and transport overhead) for the multi-attribute Foursquare encoding, combining the interaction of each client with both the randomness and the aggregation servers.
Each client submits about 300 bytes per attribute with some fixed overhead for internal framing combining all communication, with most of that traffic going to the aggregation server. See Section~\ref{sec:comm-analysis} for a detailed breakdown.
Dummy data submitted to hide below-threshold submissions is $t \frac{(\tau - 1)\tau}{2}$ (expectation) and $2t\frac{(\tau - 1)\tau}{2}$ (worst-case) where $t$ is the expectation of the truncated discrete Laplace distribution and $\tau$ is the threshold for pruning values. Given our parameter values $t=14$ and $\tau=20$, this amounts to a few thousand extra aggregation server reports which is still quite small in absolute terms, and completely negligible on the server side.
Therefore, \textbf{\name introduces very small bandwidth costs for clients}.

\begin{figure}[t!]
\centering
\small
\setlength{\tabcolsep}{1pt}
\begin{tabular}{ll}
\begin{tikzpicture}
\begin{axis}[cycle list name=color list,
                  axis y line*=left, 
                  width=4.5cm,
                  height=4.5cm,
                  title= \small Bandwidth Overhead (MB),
                  enlarge y limits=0.01,
                  enlarge x limits=0.02,
                  xlabel={\small Number of attributes},
                  ylabel={\small Client \& Randomness Server},
                  tick label style={/pgf/number format/fixed},
                  scaled ticks=false,
                  ytick={0.00,0.01,0.02,0.03},
                  yticklabels={0.00,0.01,0.02,0.03}
                  ]

\addplot+[line width=0.8pt,densely dashdotted]
              table[x=x,y=y1] {scale_bandwidth_cr.txt};
\end{axis}
\begin{axis}[cycle list name=color list,
                  axis lines*=right, 
                  axis x line=none,
                  width=4.5cm,
                  height=4.5cm,
                  %title= \small Client \& Aggregation Server,
                  enlarge y limits=0.01,
                  enlarge x limits=0.02,
                  ylabel={\small Client \& Aggregation Server},
                  ytick={0.00,0.10,0.20,0.30},
                  yticklabels={0.00,0.10,0.20,0.30}
                  ]

\addplot+[line width=0.8pt,densely dashdotted]
              table[x=x,y=y2] {scale_bandwidth_cr.txt};

\end{axis}
\end{tikzpicture}
\begin{tikzpicture}
\begin{axis}[cycle list name=color list,
axis lines=left, 
                  width=4.5cm,
                  height=4.5cm,
title= Running time (Seconds),
                  enlarge y limits=0.01,
                  enlarge x limits=0.02,
                  xlabel={Number of attributes},
                  ylabel={Client and both Servers},
]

\addplot+[line width=0.8pt,densely dashdotted]
              table[x=x,y=y3] {scale_bandwidth_cr.txt};

\end{axis}
\end{tikzpicture}
\end{tabular}
\caption{Scalability of \name in terms of the computational and bandwidth overhead introduced for clients. \textbf{\name scales to a large number of attributes with negligible costs.}}
\label{fig:ScaleAttr}
\end{figure}

\begin{figure}[t!]
\centering
\small
   \begin{tikzpicture}
    \pgfplotsset{
y axis style/.style={
            yticklabel style=#1,
            ylabel style=#1,
            y axis line style=#1,
            ytick style=#1
       }
    }
    
    \begin{axis}[
      axis y line*=left,
width=7cm,
       height=4.5cm,
       xtick={1,5,10,15,20,25,30,33},
       xticklabels={1,5,10,15,20,25,30,33},
      xlabel=Number of clients (M),
      ylabel=Running time (seconds),
      tick label style={/pgf/number format/fixed},
    ]
    \addplot[smooth,mark=x] 
      coordinates{
        (1,6)
        (5,36) 
        (10,82)
        (15,134)
        (20,188) 
        (25,245) 
        (30,307)
        (33,344)
    };
    \end{axis}
    
    \begin{axis}[
      axis y line*=right,
      axis x line=none,
      width=7cm,
      height=4.5cm,
ylabel=Financial costs (USD),
tick label style={/pgf/number format/fixed},
    ]
    \addplot[smooth,mark=*] 
      coordinates{
        (1,0.001461)
        (5,0.009177) 
        (10,0.020700)
        (15,0.033811)
        (20,0.047355) 
        (25,0.061801)
        (30,0.077361)
        (33,0.086850)
    };
    \end{axis}
    
    \end{tikzpicture}
    \caption{Scalability of \name in terms of the computational and financial costs for the aggregation server. \textbf{\name scales to a large number of clients with negligible costs.}}
    \vspace{-0.4cm}
    \label{fig:sc-as}
\end{figure}

\looseness=-1 \myparagraph{Financial costs.} We compute the financial costs of running \name for the aggregation and randomness servers.
We benchmarked \name on an AWS r6a.4xlarge instance, currently priced at US\$0.9072 per hour. As such, the amortized cost of aggregating submissions from the IPUMS dataset is 0.03 USD, the Foursquare dataset 0.09 USD, and the Complete Works of Shakespeare only 0.002 USD.

\subsection{\name is Scalable}
\looseness=-1 We further analyze the scalability of \name by considering various numbers of attributes using the same hardware described in Section~\ref{sec:expEffic}.
Timings for the IPUMS and Foursquare datasets report the more expensive multi-attribute encoding scheme (Algorithm~\ref{alg:client-nested-encoding}), while the Shakespeare dataset does not use multiple attributes and represents the whole-value reporting scheme in general.
Figure~\ref{fig:ScaleAttr} shows the effect of the number of attributes on the bandwidth and computational costs of clients needed to interact with the randomness server and prepare the submission to the aggregation server. Interaction with the randomness server scales linearly with the number of attributes, since the OPRF must be evaluated separately for each prefix.
Bandwidth costs of each client interacting with the aggregation server also scale linearly with the number of attributes.
Finally, we observe that the time of \name run by each client scales linearly with the number of attributes. This is because run time is dominated by the key share and encryption steps which in the multi-attribute scheme need to be done once per attribute to allow partial recovery of multivariate joint attributes.

Finally, we analyzed the scalability of \name by considering various numbers of clients. Figure~\ref{fig:sc-as} shows the running time and financial costs of the aggregation server as a function of number of clients. \textbf{\name can collect multi-dimensional data from a very large number of clients with small costs.}

\section{Related work}
\label{sec:related}

\myparagraph{Threshold-aggregation data collection systems}~\cite{CCS:DSQGLH22,boneh2021lightweight,JMLR:BNST20,bassily2015local,ACMTA:BNS19,CCS:QYYKXR16,PMLR:ZKMSL20,prochlo,Chaum81,CCS:Neff01}
allow a central party to learn submitted values if and only if a predefined number of clients send the exact same value. Poplar~\cite{boneh2021lightweight} uses distributed point functions to create a secret sharing pair of a vector in which only a single element with an index corresponding to the client's data is non-zero.  Clients then send their secret shares to \emph{two non-colluding aggregation servers}, which compute the sum of submitted shares and publish the sum values. STAR~\cite{CCS:DSQGLH22} uses a different $\threshold$-out-of-$N$ secret-sharing scheme to avoid the need for two aggregation servers. Clients encode their values as secret shares, and send their shares to a \emph{single aggregation server}, which is to decrypt values that have been encoded by at least $\threshold$ submitted shares. STAR provides high efficiency and more desirable trust requirements than Poplar, though at the cost of leaking the histogram of unrecovered values (\ie{} values submitted by less than $\threshold$ clients). POPSTAR~\cite{li2024popstar} tries to hide the distribution of unrevealed values in STAR, though in a manner that requires significantly more computation (7x computation, and an estimated 2-3x increase in the required time, compared to the dataset as STAR).

Existing threshold aggregation systems have significant limitations. First, all threshold-aggregation systems, like all deterministic $k$-anonymity systems, lack robust, provable privacy guarantees. Furthermore, current threshold-aggregation systems (including Poplar and STAR) are only well-suited to handle single-dimension values. Trying to use these systems to handle multi-dimensional records entails significant utility loss. One option is to flatten multi-dimensional records into a single dimension (\eg{} concatenation, summation, etc.), and run the system on that ``flattened'' value. This approach significantly harms utility, since submitted records would need to match across all original dimensions to count towards each record's recovery threshold, decreasing the amount of information the server is able to recover. The second option is to have clients submit each dimension independently, treating a record with three attributes
as three independent records. This also harms utility, though in a different way: the aggregation server is unable to learn any relationships or correlations between data attributes.

\looseness=-1 Our proposed system, \name, works better than these systems in terms of both utility and privacy by \textit{i)} allowing clients to submit data with multiple attributes such that the utility of marginal histogram estimations is maximized; and \textit{ii)} satisfying strong differential privacy guarantees (through sampling followed by pruning, and dummy data) without trusting servers.

\myparagraph{Differentially private data collection systems.} Differential Privacy (DP)~\cite{dwork2014algorithmic} uses  
statistical indistinguishably to ensure privacy. DP is typically implemented in one of two forms: first, a \emph{central model} where the aggregation server receives unmodified user data, who then
applies privacy protections to the data before sharing it, and second, a \emph{local model}, where users apply privacy protections to their own data before sharing it with the aggregation server. These approaches achieve different privacy-utility trade offs.

\looseness=-1 Central DP systems provide high utility, but suffer from often prohibitive
trust assumptions (\ie clients must trust the aggregation server and send their raw data to the server). This is a practical problem, as many servers do not provide the privacy protections they promise (intentionally or otherwise)~\cite{tang2017privacy}. Central  DP systems also carry the risk of a single point of failure for data breaches~\cite{roy2020crypt,venturini2019api}.

Local DP systems, on the other hand, provide strong privacy
guarantees, typically by perturbing data before revealing it to untrusted parties~\cite{JMLR:BNST20, PMLR:ZKMSL20}.
This greatly improves the privacy and security properties of the system, but at the cost of reducing the utility of the aggregated data. The shuffle model of DP~\cite{balcer2019separating} improves the utility by allowing to perturb data with less noise, while trusting an intermediate shuffler to apply a uniform random permutation to all data before the aggregation server views them.

More recent DP proposals attempt to achieve better utility 
through using multi-party computation or homomorphic encryption to actually reduce the level of trust required in central  DP systems~\cite{bohler2020secure,boehler2022secure,pettai2015combining,cheu2019distributed,corrigan2017prio,bonawitz2017practical,roy2020crypt,bell2022distributed}. For example, Bell et.al.,~\cite{bell2022distributed} present a protocol for computing DP histograms carried out between two non-colluding aggregation servers. In each round of server-to-server communication, each server injects carefully crafted dummy data into its message, ensuring that the additional leakage revealed to other server beyond the output remains differentially private--a DP anonymized histogram of indices to one server and a DP anonymized histogram of values to other server. But these proposals suffer from their own drawbacks, including \textit{i)} requiring networks of non-colluding servers, a majority of whom are assumed to behave honestly~\cite{corrigan2017prio}, \textit{ii)} imposing high
computation and communication overheads costs~\cite{cheu2019distributed}, \textit{iii)} only being suited for simple aggregation functions~\cite{bonawitz2017practical,boehler2022secure,bohler2020secure,pettai2015combining}, and \textit{iv)} requiring interactive communication between clients and servers. Requiring interactions between the servers (more than 1 aggregation server) makes it more difficult to guarantee the non-collusion requirement and it may also have practical costs (e.g., it may be harder to recruit collaborative partners).
Finally, most DP systems are also limited to one-dimensional data, limiting their utility or applicability to many scenarios.
 
\looseness=-1 Our system, \name, obtains better utility than both local DP randomized and shuffling\footnote{Bharadwaj and Cormode~\cite{cormode2022sample} analytically and empirically demonstrated the superior performance of the sample-and-threshold compared to the shuffle model of DP.} as: \textit{i)} the utility error of \name is independent of the number of attributes as opposed to local DP randomizers in which the noise grows significantly as the number of attributes increases; \textit{ii)} in contrast to existing DP  systems, \name does not add explicit noise, thus introducing no spurious attribute values. In addition to this, our system avoids prohibitive trust assumptions required in the central model deployment of DP while being efficient because of not requiring expensive multi-party computations and homomorphic encryption operations. 

\myparagraph{Trusted hardware.}  Finally, a third-general approach to private data collection uses
trusted hardware to enforce privacy guarantees. For example, Prochlo~\cite{prochlo} uses trusted
hardware to collect unmodified data from clients. This trusted hardware is able to collect, shuffle, and modify user data \emph{before} privacy protections are applied to the data. Once these trusted servers have received sufficient data, it is modified and passed onto untrusted hardware, which does the primary data summarizing and aggregation. Trusted hardware carries a wide range of downsides and limitations though, including relatively high cost, resource limitations (in some cases), and (in many cases) merely re-shuffled trust requirements. Approaches like mix-nets~\cite{Chaum81} and verifiable shuffling~\cite{CCS:Neff01} can provide security and privacy guarantees similar to (but without requiring) trusted hardware, though at the cost of increased interactivity.

\section{Discussion and Future work}
\label{sec:discussion}
In this paper, we proposed \name that can be used to privately estimate and publish histogram of data generated by clients.  \name introduces necessary randomness for the privacy protection of clients through sampling, thresholding, and dummy data injection, and removes the trust assumption on the server through a customized secret-sharing protocol. Incorporating synergies and optimizations on both fronts of sample-and-threshold differential privacy and secure threshold aggregation enables \name to provide high utility without imposing prohibitive trust requirements, relying on computationally expensive cryptographic operations, or requiring bandwidth-intensive multi-round communications between clients and servers. We analytically and empirically demonstrated that \name is effective, efficient and scalable. 

We conclude by discussing some limitations of our approach and by outlining promising directions for future work. As discussed in Section~\ref{sec:problem_formulation}, our threat model, like those commonly found in the literature, assumes honest-but-curious clients. An interesting future work is to extend \name to defend against malicious clients who either deviate from the protocol or collude with one or both servers to compromise the privacy of honest clients. Another interesting direction involves designing efficient protocols for the verifiable selection of a single client (or a subset of clients) to generate and submit dummy data, as outlined in Section~\ref{sec:dummydataInject}.

\myparagraph{Acknowledgements.} We thanks Olive Franzese, Sofía Celi, and Alex Davidson for useful comments and feedback.

\bibliographystyle{IEEEtran}
\bibliography{refs.bib}
\appendix

\section{Security Proof}
\label{app:secproof}

\myparagraph{Formal Security Model}. Aggregation Server $S_A$ may be corrupted by a \emph{malicious} adversary with arbitrary behavior. It is assumed not to collude with the randomness server or any clients; Randomness Server $S_R$ may be corrupted by a \emph{malicious} adversary with arbitrary behavior. It is assumed not to collude with $S_A$ or any clients; Clients $\{C_i\}^N_{i=1}$ may be corrupted by a \emph{semi-honest} adversary. Corrupted clients follow the protocol, but seek to learn information about other parties. They are assumed not to collude with other parties.

\name effectively employs a protocol for secure computation of private threshold aggregation reporting (STAR~\cite{CCS:DSQGLH22}) as a subroutine, with additional elements for guaranteeing differential privacy over its outputs. To reason about the security of \name, we will begin by recalling the elements that are the same between the two protocols.

In particular, Algorithm~\ref{alg:client-RandomnessServer} plus lines 1-2 of Algorithm~\ref{alg:client-encoding} are equivalent to the Randomness Phase of~\cite{CCS:DSQGLH22} with a VOPRF as instantiated in~\cite{albrecht2021round}, assuming the hash function $H$ is a secure random oracle. The remainder of Algorithm~\ref{alg:client-encoding} is equivalent to the Message Phase of~\cite{CCS:DSQGLH22} for the subset of clients which are selected by the local Bernoulli tests in lines 6-8 of Algorithm~\ref{alg:mechanism1}, and for the clients that are not selected it is equivalent to aborting immediately before the Message Phase (this is tolerated trivially by the security proofs in~\cite{CCS:DSQGLH22}). Algorithm~\ref{alg:dummy-modified} similarly replicates the Randomness Phase and Message Phase of~\cite{CCS:DSQGLH22} in lines 7 and 8 respectively, followed by submission to $S_A$ in Algorithm~\ref{alg:mechanism1} line 10. Finally, Algorithm~\ref{alg:aggregation} is equivalent to the Aggregation Phase of~\cite{CCS:DSQGLH22}. 

These elements can be considered an implementation of the protocol from~\cite{CCS:DSQGLH22}, realizing the functionality $\mathcal{F}_\text{STAR}$ which we recapitulate in Figure~\ref{functionality:star}. We make one addition, formalizing the fact that their protocol enables clients to participate in the Randomness Phase and then abort before submitting data in the Message Phase. We use these observations to prove the security of \name based on the Universal Composition paradigm~\cite{canetti2000uc}, a standard cryptographic proof technique. Before we proceed to the proof, we will recall elements of STAR and \name that are \emph{dissimilar}.

The Bernoulli tests in lines 6-8 of Algorithm~\ref{alg:mechanism1} are not present in~\cite{CCS:DSQGLH22} and neither are lines 9-11 of Algorithm~\ref{alg:mechanism1}, which call Algorithm~\ref{alg:dummy} as a subroutine. As previously mentioned, the former is trivially tolerated by~\cite{CCS:DSQGLH22}. The latter is equivalent to increasing the number of client inputs received by the STAR functionality, which are chosen so that they reveal no information about the input data of any client. Thus our protocol for \name can be rewritten as: i) $S_R$ samples a VOPRF keypair (msk, mpk); ii) Clients $\{C_i\}^N_{i=1}$ sample $b_i \gets \text{Bern}(p_s)$; iii) A randomly chosen client $C^*$ samples $\{\alpha_i\}^{\threshold - 1}_{i=1}$ where each $\alpha_i$ is sampled independently from $\alpha_i \leftarrow \text{TSDLap}(\lambda=2/\varepsilon_{\text{Unre}},t=2+2/\varepsilon_{\text{Unre}}\log(2/\delta_{\text{Unre}}))$; %Let $\phi \gets \sum^{\threshold-1}_{i=1} \alpha_i$.
    % \item For each $i \in [1, \tau-1]$ construct $\alpha_i$ dummy clients each with input $(\omega_i, \text{aux}_i)$ where each $\omega_j$ is a distinct measurement outside the set of client measurements $\{x_i\}^N_{i=1}$
iv) For each $\alpha_i$, $C^*$ constructs $\alpha_i$ groups of dummy clients of size $i$ with input $(\omega_{i,j}, \text{aux}_{i,j})$ for all $j \in \alpha_i$. Each $\omega_{i,j}$ is a distinct measurement outside the set of client measurements $\{x_i\}^N_{i=1}$. Call the set of all dummy clients $\mathcal{D}$;
    %\item A randomly chosen client simulates $\phi$ clients with dummy data $\{C'_j\}^\phi_{j=1}$, where $\phi \gets \sum^{\threshold-1}_{i=1} \alpha_i$ and each $\alpha_i$ is sampled from $\alpha_i \leftarrow \text{TSDLap}(\lambda=2/\varepsilon_{\text{Unre}},t=2+2/\varepsilon_{\text{Unre}}\log(2/\delta_{\text{Unre}}))$. Each dummy client uses $b_j \gets 1$.
v) $S_A, S_R,$ and client pool $\{C_i\}^N_{i=1} \cup \mathcal{D}$ call $\mathcal{F}_\text{STAR}$ using their respective inputs from above. Assume additional leakage of the cardinality of each group of inputs $n_\ell \equiv |S_\ell|$ where $S_\ell$ is the set of submissions from $\mathcal{C} \cup \mathcal{D}$ which share the same unique measurement $x_\ell$, as in the protocol from~\cite{CCS:DSQGLH22}.

We use this formulation of our protocol to demonstrate that \name realizes the ideal functionality $\mathcal{F}_\text{Nebula}$ shown in Figure~\ref{functionality:nebula}. This augmented functionality ensures that the output histogram has the DP guarantees proven in Section~\ref{sec:analysis}. Note that we also output to $S_A$ the cardinality of each group of inputs provided by $\mathcal{C}\cup\mathcal{D}$ which shares a unique measurement $x_\ell$. This is also included in STAR, but excluded from the ideal functionality and formalized as a leakage function $\mathcal{L}$. We include it as an output of $\mathcal{F}_\text{Nebula}$ since the method is designed specifically to make this leakage uninformative using differential privacy.

\begin{figure}
\begin{tcolorbox}[title=Ideal Functionality $\mathcal{F}_\text{STAR}$]
\textbf{Participants:} 
\begin{itemize}[leftmargin=1.5em]
\item Aggregation server $S_A$ 
\item Randomness server $S_R$  
\item Clients $\{C_i\}_{i=1}^N$
\end{itemize}
\textbf{Public parameters:}
\begin{itemize}[leftmargin=1.5em]
    \item Threshold $\threshold$
\end{itemize}
\textbf{Inputs:}
\begin{itemize}[leftmargin=1.5em]
        \item $S_R$: provides VOPRF keypair $(\text{msk}, \text{mpk})$.
        \item Client $C_i \in \{C_i\}^N_{i=1}$ provides input $(x_i, \text{aux}_i)$ and a bit $b_i$ to indicate protocol abort immediately after the Randomness Phase.
        \item $S_A$ has no inputs, provides $\bot$
\end{itemize}
\textbf{Functionality:}
\begin{enumerate}
    \item For each unique $x_\ell$ construct:
    \[
    \mathcal{E}_\ell = \left\{(x_\ell, \{\text{aux}_j\}_{j \in J}, \threshold_\ell) : \left(J \subseteq [N]\right) \wedge \left(x_j = x_\ell\right) \wedge (b_j = 1) \right\}
    \]
    where $\threshold_\ell = \left|\{\text{x}_j\}_{j \in J}\right|$ is the number of sampled client measurements in $\mathcal{E}_\ell$.
    \item Let $\mathcal{Y}$ be an empty map.
    \item For each $\mathcal{E}_\ell$ where $\threshold_\ell \geq \threshold$, set $\mathcal{Y}[x_\ell] = \mathcal{E}_\ell$.
\end{enumerate}
\textbf{Outputs:}
\begin{itemize}[leftmargin=1.5em] 
        \item output $\mathcal{Y}$ to $S_A$
        \item output $\{\mathcal{F}_\Gamma(\text{msk}, x_i)\}^N_{i=1}$ to $S_R$, where $\mathcal{F}_\Gamma$ is an ideal functionality for VOPRF protocol $\Gamma$. 
        \item output $\bot$ to each $\{C_i\}^N_{i=1}$
\end{itemize}
\end{tcolorbox}
\caption{Ideal functionality for STAR~\cite{CCS:DSQGLH22}.}\label{functionality:star}
\vspace{-0.4cm}
\end{figure}

Now ready to prove the security and correctness of our protocol.

\begin{algorithm2e}[t!]
\algsetup{linenosize=\tiny}
\small
\DontPrintSemicolon
\SetKwComment{Comment}{{\scriptsize$\triangleright$\ }}{$\quad\quad$}
\caption{\textit{Modified DummyDataCreation}: Create groups of dummy data}\label{alg:dummy-modified}
        \KwIn{Threshold $\threshold$, Truncated Shifted Discrete Laplace distribution $\text{TSDLap}(\cdot)$, DP guarantees $(\varepsilon_{\text{Unre}},\delta_{\text{Unre}})$}
        \KwOut{A set of dummy data}
\BlankLine
\begin{algorithmic}[1] 
    \STATE $\text{Dummy}=\{\}$ 
    \STATE Select a client for creating dummy data
    \FOR{$i = 1,\ldots,\threshold-1$}
    \STATE $\alpha \leftarrow \text{TSDLap}(\lambda=2/\varepsilon_{\text{Unre}},t=2+2/\varepsilon_{\text{Unre}}\log(2/\delta_{\text{Unre}}))$ 
    \FOR{$j \in \alpha$}
    \STATE $x_j \gets \text{DummyObservation()}$
    \STATE $r_j = ClientRandomnessServer(x_j, pp, H(\cdot))$
    \STATE $\text{sbm}_j$, \_ $\gets LocalSecretSharing(x_j, r_j, \Pi_{\threshold, N})$
    \STATE $\text{Dummy}.\text{append}(\{(\text{sbm}_j)^i\})$
    \ENDFOR
    \ENDFOR
    \STATE \textbf{Return} $\text{Dummy}$
\end{algorithmic}
\end{algorithm2e}
\begin{theorem}
    (Correctness.) The \name protocol is correct with all but negligible probability.
\end{theorem}
\begin{proof}
Clients perform Bernoulli tests to decide whether to submit their real data (Algorithm~\ref{alg:mechanism1}), and craft and submit dummy data (Algorithm~\ref{alg:dummy-modified}). Therefore, the set of data received by the aggregation server in the protocol is equivalent to the inputs specified in $\mathcal{F}_\text{Nebula}$. The remaining parts of the protocol reduce to performing the STAR protocol on this modified set of data (Poisson subsampled real data joint with dummy data), resulting in the aggregation server learning only a DP histogram of revealed submissions and a DP histogram of unrevealed submissions as proven in our Theorem~\ref{theorem:DP}. Thus the rest of the proof reduces to the correctness proof for STAR, which demonstrates correctness except with negligible probability (\cite{CCS:DSQGLH22} Theorem 2). Therefore, \name correctly produces the outputs shown in $\mathcal{F}_\text{Nebula}$ (i.e., revealing only the expected DP histogram to the aggregation server as shown in Figure~\ref{functionality:nebula}).
\end{proof}

\begin{theorem}
    (Malicious aggregation server.) The \name protocol is secure against any $\mathcal{A}$ that corrupts $S_A$, assuming a secure protocol which realizes $\mathcal{F}_\text{STAR}$.
\end{theorem}
\begin{proof}
    Given any $\mathcal{A}$, we define a simulator $\mathcal{S}$ which interacts with $\mathcal{A}$ and the ideal functionality $\mathcal{F}_\text{Nebula}$ as follows: $\mathcal{S}$ submits $\bot$ to $\mathcal{F}_\text{Nebula}$; $\mathcal{S}$ receives from $\mathcal{F}_\text{Nebula}$ the histogram $\mathcal{Y}$, and cardinality of each group of inputs $n_\ell$ for every unique measurement $x_\ell$; $\mathcal{S}$ simulates $\mathcal{F}_\text{STAR}$: when $\mathcal{A}$ submits $\bot$ to $\mathcal{F}_\text{STAR}$, $\mathcal{S}$ sends $\mathcal{Y}$ back; $\mathcal{S}$ simulates the leakage $\mathcal{L}$ by submitting each $n_\ell$ to $\mathcal{A}$.

$\mathcal{S}$ simulates the view of $\mathcal{A}$. Uncorrupted clients correctly perform Poisson sampling and dummy data injection locally in the real-world execution, resulting in a set of outputs to $\mathcal{F}_\text{STAR}$ which is sampled from the same distribution as the outputs of $\mathcal{F}_\text{Nebula}$. Since our security model specifies that clients do not collude with $S_A$, we can assume all clients are uncorrupted when $\mathcal{A}$ corrupts $S_A$. So the simulated and real-world views are statistically indistinguishable. 
\end{proof}
Operations conducted between the clients and the randomness server $S_R$ are the same as in STAR~\cite{CCS:DSQGLH22}. 
\begin{theorem}
    (Malicious randomness server.) The \name protocol is secure against any adversary $\mathcal{A}$ that corrupts $S_R$, assuming a secure protocol which realizes $\mathcal{F}_\text{STAR}$
\end{theorem}
\begin{proof}
    Given any $\mathcal{A}$, we define a simulator $\mathcal{S}$ which interacts with $\mathcal{A}$ and the ideal functionality $\mathcal{F}_\text{Nebula}$ as follows:  simulating $\mathcal{F}_\text{STAR}$, $\mathcal{S}$ receives inputs for (msk, mpk) from $\mathcal{A}$; $\mathcal{S}$ submits (msk, mpk) to $\mathcal{F}_\text{Nebula}$ and receives $\{\mathcal{F}_\text{VOPRF}(\text{msk}, x_i)\}^N_{i=1}$; $\mathcal{S}$ uses (msk, mpk) to simulate $\mathcal{F}_\text{VOPRF}(\text{msk}, x_j)$ for all dummy clients in $\mathcal{D}$. Call this set of outputs $W$; $\mathcal{S}$ sends $\{\mathcal{F}_\text{VOPRF}(\text{msk}, x_i)\}^N_{i=1} \cup W$ to $\mathcal{A}$ as the output of $\mathcal{F}_\text{STAR}$.
    This simulates the view of $\mathcal{A}$.
\end{proof}

\begin{theorem}
    (Semi-honest corrupted clients.) The \name protocol is secure against any adversary $\mathcal{A}$ that corrupts a subset of clients $T \subset \{C_i\}^N_{i=1}$, assuming a secure protocol which realizes $\mathcal{F}_\text{STAR}$.
\end{theorem}
\begin{proof}
    The simulator takes inputs submitted to $\mathcal{F}_\text{STAR}$ and submits them to $\mathcal{F}_\text{Nebula}$. It then returns $\bot$ to each client in $T$. This trivially simulates the view of $\mathcal{A}$.
    
    To show that the joint distribution over inputs and outputs is statistically indistinguishable for real and ideal world execution, we recall that clients are corrupted only by semi-honest adversaries in our security model. Thus we can assume that all corrupted clients in $T$ follow the protocol faithfully. This means that the clients perform Poisson sampling and dummy data injection correctly even if they are corrupted in the real-world execution. Thus, the joint distributions in the real and ideal world executions are the same.
\end{proof}

\end{document}